\documentclass[english,11pt]{article}
\usepackage[utf8]{inputenc}
\usepackage{amsmath,amssymb,amsthm}
\usepackage{mathrsfs}
\usepackage{graphicx}
\usepackage{appendix}
\usepackage{fullpage}
\usepackage{parskip}
\usepackage{setspace}
\usepackage{multirow}
\usepackage{multicol}
\usepackage{booktabs}
\usepackage{bm}
\usepackage{natbib}
\usepackage{thmtools}
\usepackage{enumitem}
\usepackage{hyperref}
\usepackage{color}
\usepackage[dvipsnames]{xcolor}
\definecolor{darkblue}{rgb}{0, 0, 0.5}
\hypersetup{
    colorlinks = true,
    citecolor = darkblue,
    urlcolor = darkblue,
    linkcolor = darkblue
}
\usepackage{cleveref}
\crefname{thm}{Theorem}{Theorems}
\crefname{assumption}{Assumption}{Assumptions}
\crefname{subassumptioni}{Assumption}{Assumptions}
\crefname{subappendix}{Appendix}{Appendices}
\Crefname{subappendix}{Appendix}{Appendices}

\usepackage{algorithm}
\usepackage{algpseudocode}
\usepackage{authblk}

\newlist{subassumption}{enumerate}{1}
\setlist[subassumption,1]{
  label=(\roman*),   
  ref=\theassumption(\roman*), 
  leftmargin=2em
}

\newcommand{\E}{\operatorname{\mathbb{E}}}
\newcommand{\V}{\operatorname{\mathbb{V}}}

\newcommand{\supp}{\operatorname{\text{supp}}}
\newcommand{\R}{\mathbb{R}}
\newcommand{\B}{\mathbb{B}}
\newcommand{\N}{\mathbb{N}}

\newcommand{\G}{\mathbb{G}}
\newcommand{\X}{\mathcal{X}}
\newcommand{\Y}{\mathcal{Y}}
\newcommand{\F}{\mathscr{F}}

\newcommand{\weakto}{\rightsquigarrow}
\newcommand{\bydef}{:=}

\newtheorem{theorem}{Theorem}
\newtheorem{corollary}{Corollary}
\newtheorem{proposition}{Proposition}

\newtheorem{assumption}{Assumption}
\makeatletter
\@addtoreset{subassumptioni}{assumption}
\let\cref@OT@resetby\cref@resetby
\def\cref@resetby#1#2{%
  \let#2\relax%
  \cref@ifstreq{#1}{subassumptioni}{%
    \def#2{assumption}%
  }{%
    \cref@OT@resetby{#1}{#2}%
  }%
}
\makeatother
\newtheorem{remark}{Remark}
  
\title{\textbf{Inference in partially identified moment models\\ via regularized optimal transport}\footnote{
    We are grateful to Tim Armstrong, Yanqin Fan, Arie Kapteyn, Toru Kitagawa, Chen Qiu, seminar participants at the University of Copenhagen and New Economic School, and conference participants at the California Econometrics Conference 2025, Aarhus Workshop in Econometrics VII, and CEME Young Econometricians Conference 2025.
    The project described in this publication relies on data from surveys administered by the Understanding America Study (UAS), which is maintained by the Center for Economic and Social Research (CESR) at the University of Southern California. The project was supported by the National Institute on Aging of the National Institutes of Health and, in part, by the Social Security Administration under Award Number U01AG077280. The content is solely the responsibility of the authors and does not necessarily represent the official views of the National Institutes of Health, USC, or UAS.
}}

\date{\today}

\author[1]{\textsc{Grigory Franguridi}}
\author[2]{\textsc{Laura Liu}}

\affil[1]{Center for Economic and Social Research, University of Southern California  \vspace{1ex}}
\affil[2]{Department of Economics, University of Pittsburgh}

\begin{document}

\maketitle

\begin{abstract}
\setstretch{1.2}

In many economic analyses, parameters of interest are functionals of a joint distribution that is identified only up to its marginals.
We develop a framework for identification, estimation, and inference in the resulting partially identified GMM model.
We characterize the sharp identified set for the parameter of interest using optimal transport (OT) and construct a computationally efficient estimator based on \emph{entropic regularization}.
We also propose a statistic for testing hypotheses and constructing confidence regions.
To derive its asymptotic distribution, we establish a central limit theorem for the entropic OT value under general smooth cost functions.
We obtain valid critical values using the bootstrap for directionally differentiable functionals of \citet{fang2019inference}.
The resulting test controls size locally uniformly, including at the boundary of the identified set.
Although regularization induces conservativeness, we show how to limit it through the choice of the regularization parameter.
We show in Monte Carlo simulations that our estimator and confidence regions perform well in finite samples.
Finally, we illustrate our methodology by estimating a model of self-reported happiness using data from the Understanding America Study (UAS), allowing for potentially nonrandom attrition.

\bigskip


\bigskip


\end{abstract}

\onehalfspacing

\newpage

\section{Introduction}

Many quantities of interest in economics are not point-identified under realistic modeling choices. Two broad reasons can account for this: incomplete data, where relevant variables are only partially observed, and incomplete models, where economic theory does not pin down a unique data-generating mechanism. In such settings, the available data and modeling assumptions generally imply a set of parameter values rather than a unique point. The partial identification literature, beginning with the pioneering work of \citet{manski1990nonparametric}, studies what can be learned about economically relevant parameters when the available information is insufficient for point identification; see also the handbook of \citet{manski2003partial} and the survey of \citet{tamer2010partial}. 

The focus of this paper is a model in which parameters are defined by moment restrictions involving a joint distribution that is identified only up to its marginals.
Specifically, we consider a generalized method of moments (GMM) model where the parameter $\theta_0 \in \Theta \subset \R^k$ satisfies moment conditions
\begin{align}
    \E_{\pi_0}[\phi(X,Y,\theta_0)] = 0, \label{eq:gmm-intro}
\end{align}
but the joint distribution $\pi_0$ of $(X,Y)$ is only known to lie in the set of couplings $\Pi(\mu,\nu)$ with identified marginals $\mu$ and $\nu$. 
To characterize the sharp identified set, fix a candidate $\theta$ and consider the set
\[
\mathcal{M}(\theta) \bydef \left\{\E_\pi[\phi(X,Y,\theta)] : \pi \in \Pi(\mu,\nu)\right\}
\]
of all moment values generated by couplings consistent with the identified marginals. Then $\theta$ is compatible with the data if and only if $0 \in \mathcal{M}(\theta)$, or equivalently, the distance to the origin
\begin{align}
d\left(0,\mathcal{M}(\theta)\right)
\bydef \min_{z \in \mathcal{M}(\theta)} \|z\| = \min_{\pi \in \Pi(\mu,\nu)} \left\| \E_{\pi}\left[\phi(X,Y,\theta)\right] \right\|
= \max_{\|u\|\le 1}\min_{\pi \in \Pi(\mu,\nu)} \E_{\pi}\left[u'\phi(X,Y,\theta)\right]
\label{eq:dist-intro}
\end{align}
is equal to zero. The last equality follows by norm duality and Sion's minimax theorem: see \Cref{sec:model-and-sharp-identified-set} for details.
The inner minimization is an optimal transport (OT) problem with the cost function $u'\phi(x,y,\theta)$, whose solution is the coupling that minimizes the moment violation in direction $u$.
The outer maximization identifies the worst-case direction.
Hence, $\theta$ belongs to the sharp identified set exactly when the max-min value in \eqref{eq:dist-intro} equals zero.

As a toy example, consider a randomized controlled trial (RCT) with potential outcomes $Y(0) \sim \mu$ and $Y(1) \sim \nu$.
The share of units that benefit from treatment
\(
\theta = \mathbb P_\pi\left(Y(1) \ge Y(0)\right)
\)
is not point-identified since only the marginals $\mu,\nu$ of $\pi$ are observed. The moment function
\(
\phi(Y(0),Y(1),\theta)=\mathbf{1}\{Y(1) \ge Y(0)\} - \theta
\) is one-dimensional, and $\theta$ belongs to the sharp identified set $\Theta_{I}$ if and only if there exists a coupling $\pi \in \Pi(\mu,\nu)$ such that
\(
\E_{\pi}\left[\phi(Y(0),Y(1),\theta)\right] = 0. 
\)
Equivalently, $\Theta_{I} = [\underline{\theta}, \overline{\theta}]$, where $\underline{\theta} = \min_{\pi \in \Pi(\mu,\nu)} \mathbb P_\pi\left(Y(1) \ge Y(0)\right)$ and $\overline{\theta} = \max_{\pi \in \Pi(\mu,\nu)} \mathbb P_\pi\left(Y(1) \ge Y(0)\right)$. This logic also extends to vector-valued or implicitly defined parameters.

In this paper, we develop a complete methodology for identification, estimation, and inference in the OT-based partially identified GMM model \eqref{eq:gmm-intro}.

First, we characterize the sharp identified set for the parameter of interest using OT, as discussed above. This characterization has a geometric interpretation via support functions and informs an estimation procedure for the identified set.

Second, the implied estimation procedure involves solving an empirical version of the classical OT problem, which is known to be sensitive to sampling noise and to be computationally intensive with large datasets. To overcome this challenge, we employ \emph{entropic regularization}, which penalizes the negative entropy of the joint distribution. Entropic OT has been widely used in statistics and machine learning since the seminal works by \citet{cuturi2013sinkhorn} and \citet{galichon2022cupid}, see the literature review below. This regularization admits fast computation via the \emph{Sinkhorn algorithm} for the inner (entropic OT) problem and the projected gradient ascent for the outer problem in the max-min representation \eqref{eq:dist-intro}. 

Third, we develop a procedure to test hypotheses and construct confidence regions for the parameter of interest. To this end, we establish a central limit theorem (CLT) for the entropic OT value that is uniform over directions $u\in \B$ and parameter values $\theta\in\Theta$. 
We build on \citet{mena2019statistical} and \citet{goldfeld2024statistical} and extend their results to accommodate a broad class of smooth cost functions. To the best of our knowledge, we are the first to establish a CLT of such generality for the entropic OT value.

We then apply the functional delta method to the $\max$ functional to obtain the asymptotic distribution of our test statistic, similarly to \citet{franguridi2025set}. This distribution depends on the unknown argmax set and is not available in closed form. Moreover, the standard bootstrap may be invalid when the hypothesized parameter value is on the boundary of the identified set, which occurs when the argmax set is not a singleton. We resolve this issue by employing the bootstrap for directionally differentiable functionals of \citet{fang2019inference}. The resulting bootstrap-based test controls size locally uniformly and can be inverted to obtain confidence regions. Although the test is inherently conservative due to regularization, we show how to limit this conservativeness through the choice of the regularization parameter.


Finally, our methodology can be used in a variety of settings where marginal distributions are identified, whereas the joint distribution is not. In \Cref{sec:examples}, we discuss three examples: fixed effects panel logit with attrition and refreshment samples, nonparametric instrumental variables (IV) without a large support condition, and the Euler equation with repeated cross-sections. In the first example, we exploit the panel structure by fixing the joint distribution of units who stay in the sample and solving OT only for those who drop out. This tightens the bounds for the common slope coefficient and average marginal effects (AME).

\paragraph{Related literature.}
First, our paper contributes to the extensive literature on partially identified models.
Among others, \citet{imbens2004confidence} and \citet{stoye2009more} study confidence regions for partially identified parameters with uniform coverage and optimality properties, \citet{beresteanu2008asymptotic} and \citet{bontemps2012set} develop asymptotic theory and geometric characterizations for partially identified models using random sets and support functions, and \citet{romano2010inference} provide general procedures for inference on identified sets defined via minimization of a criterion function. For an overview of the broader literature on partial identification and inference, we refer to the review by \citet{canay2017practical}.
The closest antecedent is \citet{beresteanu2011sharp}, which characterizes sharp identification regions in models with convex moment predictions using the theory of random sets and support-function representations. Our characterization of the identified set relies on support functions but not on representations via random sets.

Second, a growing literature applies OT methods to economic problems. The paper closest to ours is \citet{fan2025partial}, which studies partial identification of a finite-dimensional parameter defined by a moment equality model with incomplete data via a classical OT representation of the identified set. In contrast, our paper employs entropic OT and goes beyond identification by developing the methodology for estimation and inference.
Our strategy for characterizing the identified set via an OT-based criterion function is similar to the one introduced by \citet{ekeland2010optimal,galichon2021inference}.
More broadly, OT has been deployed in a variety of economic applications, such as discrete choice models \citep{chiong2016duality}, covariate matching for causal effects \citep{gunsilius2021matching}, nonlinear difference-in-differences for multivariate counterfactuals \citep{torous2024optimal}, policy learning in matching markets \citep{hazard2025whom}, and combining stated and revealed preferences \citep{meango2025combining}.
Other works include \citet{voronin2025generalized}, which introduces a generalized version of OT for estimation in a large class of partially identified models, and \citet{schennach2025optimally}, which uses OT for estimation and inference in the overidentified GMM with measurement errors.

Third, we employ entropic regularization for the classical OT, which yields a strictly convex program that can be solved efficiently by the Sinkhorn algorithm.
Entropic OT is now widely used in high-dimensional statistics and machine learning.
In seminal works, \citet{cuturi2013sinkhorn} introduces Sinkhorn distances as a fast approximation to Wasserstein distances based on entropic regularization, and \citet{galichon2022cupid} show how entropic regularization of social surplus yields efficient algorithms for estimating matching models.
We also contribute to the statistical theory of the entropic OT and rely on two important prior works.
\citet{mena2019statistical} establish asymptotic normality of the entropic OT cost for quadratic cost functions, and \citet{goldfeld2024statistical} extend this analysis using empirical process theory to derive semiparametric efficiency bounds and bootstrap inference procedures. Our CLT extends these results by allowing for a broad class of smooth cost functions, which is essential when the cost itself encodes economically meaningful restrictions.


Finally, our analysis for the panel logit with attrition and refreshment is closely related to recent contributions on the analysis of panel data with nonignorable attrition and refreshment in \citet{franguridi2024estimation,franguridi2024robust,franguridi2024closed}. These papers develop computationally feasible, robust procedures for estimation and inference in this setting under separability assumptions on attrition that guarantee point identification. In contrast, we do not impose any assumptions on attrition, and hence our model is partially identified even when a refreshment sample is available. We then show how to characterize and estimate bounds on the common slope parameter and the AME \citep[see, e.g.][]{davezies2021identification}, as well as conduct inference in this setting. 

The remainder of the paper is organized as follows.
\Cref{sec:model} introduces the partially identified OT-based GMM model, and characterizes the sharp identified set for the parameter of interest.
\Cref{sec:estimation-and-inference} develops procedures for estimation and inference using entropic regularization.
\Cref{sec:examples} discusses several economic models that fit our general framework.
\Cref{sec:simulations} conducts a Monte Carlo simulation for the fixed effects panel logit with attrition and refreshment.
\Cref{sec:empirical} presents an empirical illustration of our methodology using the data from the Understanding America Study.
\Cref{sec:conclusion} concludes.
\Cref{app:proofs} contains proofs of all the theoretical results and \Cref{app:logit-details} provides additional details for the fixed effects panel logit with attrition and refreshment.
\section{Setup and partial identification}\label{sec:model}

\subsection{Model and sharp identified set}\label{sec:model-and-sharp-identified-set}

Let $X$ and $Y$ be random vectors in $\R^d$ with the joint distribution $\pi_0$ identified only up to the class $\Pi(\mu,\nu)$ of joint distributions with marginals $\mu$ and $\nu$.
The true value $\theta_0$ of a parameter of interest $\theta\in\Theta \subset \R^k$ uniquely satisfies the set of moment conditions
\begin{align*}
  \E_{\pi_0}[\phi(X,Y,\theta_0)] = 0,
\end{align*}
where $\phi=(\phi_1,\dots,\phi_p)'$ is a $p$-dimensional moment function.
Since $\pi_0$ is only partially identified, so is the parameter $\theta$ in general.
We allow for arbitrary $\dim(\phi)=p$ and $\dim(\theta)=k$ as long as the identified set is nonempty and compact. For $p>k$, the additional moments introduce extra directions to detect violations and may tighten the identified set.

For each fixed parameter value $\theta$, consider the set of \emph{moment predictions}
\[
\mathcal{M}(\theta) = \left\{\E_{\pi} \phi(X,Y,\theta): \pi\in\Pi \right\},
\]
i.e., the set of all moment vectors consistent with the identified set of distributions $\Pi = \Pi(\mu,\nu)$. Because expectation is linear and $\Pi$ is convex, $\mathcal{M}(\theta)$ is convex. The sharp identified set is then
\[
\Theta_{I} = \left\{\theta\in\Theta: 0\in\mathcal{M}(\theta)\right\}
= \left\{\theta\in\Theta: D(\theta)=0\right\},
\]
where $D(\theta)\bydef d(0,\mathcal{M}(\theta))$ denotes the Euclidean distance from the origin to the set $\mathcal{M}(\theta)$. Hence, identification is equivalent to the origin lying in the set of moment predictions \citep{beresteanu2011sharp}.
Note that, although $\mathcal{M}(\theta)$ is convex for each $\theta$, the identified set $\Theta_{I}$ need not itself be convex.

Let $\B$ be the unit ball in $\R^p$. By norm duality and Sion's minimax theorem, the distance $D(\theta)$ admits the following max-min representation
\begin{align*}
D(\theta) = d\left(0,\mathcal{M}(\theta)\right) & = \min_{\pi\in\Pi} \left\|\E_{\pi}\phi(X,Y,\theta)\right\|_2 = \min_{\pi\in\Pi} \max_{u\in \B} \E_{\pi} \left[u'\phi(X,Y,\theta)\right]\\
&= \max_{u\in \B} \min_{\pi\in\Pi} \E_{\pi} \left[u'\phi(X,Y,\theta)\right]
=: \max_{u\in \B} c_{\theta}(u).
\end{align*}
The inner minimization over couplings $\pi$ is an OT problem whose cost is the directional moment $u'\phi$. The outer maximization selects the direction $u$ in which the model slackness is largest. Since $D(\theta) = \max_{u\in \B}c_{\theta}(u)$ and $\theta \in \Theta_{I}$ if and only if $D(\theta) = 0$, we have that $\theta$ belongs to the sharp identified set exactly when $\max_{u\in \B}c_{\theta}(u) = 0$, or equivalently, when $c_{\theta}(u) \le 0$ for all $u \in \B$. This max-min representation informs our estimation and inference procedure based on estimating $c_{\theta}(u)$ for $u\in \B$ and checking whether the maximum is close to zero: see \Cref{sec:estimation-and-inference}.

\begin{remark}[Geometric interpretation of $c_{\theta}(u)$ and $D(\theta)$]\label{rem:geometry}
Since $\mathcal{M}(\theta)$ is convex, $\theta\in\Theta_{I}$ (i.e., $0\in\mathcal{M}(\theta)$) holds if and only if no direction $u\in\B$ separates the origin from $\mathcal{M}(\theta)$, which is equivalent to
\[
\min_{\pi \in \Pi(\mu,\nu)} \E_{\pi}\left[u'\phi(X,Y,\theta)\right] \le 0,
\quad \forall\, u \in \B.
\]
In terms of the OT value function $c_{\theta}(u)$, this is exactly the condition $c_{\theta}(u) \le 0$ for all $u \in \B$. For a different use of separating hyperplanes for characterization of identified sets, see \citet{botosaru2024adversarial}.
Moreover, by convex duality, the negative distance admits the representation
\[
-D(\theta)
= \min_{u\in \B} \max_{\pi\in\Pi(\mu,\nu)} \E_{\pi}\left[u'\phi(X,Y,\theta)\right].
= \min_{u\in \B} \max_{\pi\in\Pi(\mu,\nu)}u'\E_{\pi}\phi(X,Y,\theta),
\]
where the expression under $\min_{u\in \B}$ can be recognized as the support function of the convex set $\mathcal{M}(\theta)$ evaluated at $u$.
These identities clarify the connection to models with convex moment predictions in \citet{beresteanu2011sharp}, while highlighting the key distinction that in our setup the inner problem that defines $c_{\theta}(u)$ is infinite-dimensional.
\end{remark}

Note that our setup is not a standard moment inequality model because the inner minimization over $\pi$ depends on the direction $u$ and delivers a direction-dependent evaluation of the set of predicted moments rather than a fixed collection of inequalities. It is also different from intersection bounds that aggregate separate scalar constraints. 

To make the informal derivation above rigorous, we impose the following mild assumptions to ensure that the sharp identified set is nonempty and compact. Both of these properties are also crucial for the Hausdorff consistency of the associated estimator, see \Cref{thm:consistency}.

\begin{assumption}\label{a:id-set-characterization}
  \begin{subassumption}
    \item \label{a:param-compact} The parameter space $\Theta\subset \R^k$ is nonempty and compact.
    \item \label{a:supports-compact} The distributions $\mu$ and $\nu$ are concentrated on compact, convex sets $\X \subset \R^d$ and $\Y \subset \R^d$, respectively. 
    \item \label{a:correct-specification} The identified set $\Theta_{I}$ is nonempty, i.e., there exists $\theta_0 \in \Theta$ and $\pi_0 \in \Pi(\mu,\nu)$ such that $\E_{\pi_0}\left[\phi(X,Y,\theta_0)\right] = 0$.
    \item  \label{a:phi-continuous} For each $\theta\in\Theta$, the function $(x,y) \mapsto \phi(x,y,\theta)$ is continuous.
    \item \label{a:E-phi-continuous} For each $\pi\in \Pi(\mu,\nu)$, the function $\theta \mapsto \E_\pi [\phi(X,Y,\theta)]$ is continuous.
  \end{subassumption}
\end{assumption}
We emphasize that \Cref{a:supports-compact} allows the random vectors $X \sim \mu$ and $Y \sim \nu$ to have arbitrary distributions with bounded supports, including having both discrete and continuous components.
Also, \Cref{a:phi-continuous} can be relaxed to continuity of $\phi$ only in the continuous components of $X$ and $Y$; see the discussion after \Cref{a:UCLT}.

\begin{theorem}[characterization of identified set]\label{thm:id-set-characterization} Suppose \Cref{a:id-set-characterization} holds. Then the identified set $\Theta_{I}$ is nonempty and compact, and
  \(
  \Theta_{I} = \{\theta\in\Theta: \, D(\theta)=0\}.
  \)
\end{theorem}

\begin{proof}
  See Appendix \ref{app:id-set-characterization}.
\end{proof}

\subsection{Toy example}

To illustrate the characterization above, let us revisit the toy example in the introduction. Consider an RCT with potential outcomes $Y(0)\sim N(0,1)$ and $Y(1)\sim N(2,1)$. The parameter of interest is the share of units that benefit from treatment,
\[
\theta=\mathbb P_\pi\left(Y(1)\ge Y(0)\right)=\E_\pi\left[\mathbf{1}\{Y(1)\ge Y(0)\}\right].
\]
The corresponding scalar moment condition is
\(
\E_\pi\left[\phi\left(Y(0),Y(1),\theta\right)\right]=0,
\) where $\phi\left(Y(0),Y(1),\theta\right)=\mathbf{1}\{Y(1)\ge Y(0)\}-\theta$.\footnote{We use this example only to illustrate the geometry of the OT characterization. Our formal theoretical results impose continuity/smoothness conditions on the moment function, which rule out the indicator function on continuous supports.}
The identified set is then
\[
\Theta_{I} \bydef \left\{\theta\in\mathbb{R}:\ \max_{u\in[-1,1]} c_{\theta}(u)=0\right\},\quad\text{where } c_{\theta}(u)=\min_{\pi\in\Pi} \E_{\pi} \left[u\phi(Y(0),Y(1),\theta)\right].
\]
This characterization has a simple interpretation as a two-sided game. Since $\phi$ is scalar, $u$ simply flips the sign of the moment: $u=1$ tests whether $\E_\pi[\phi\left(Y(0),Y(1),\theta\right)]$ is positive under every $\pi\in\Pi$, while $u=-1$ tests the opposite. The adversary chooses the least favorable coupling under each sign. Thus,
\[
\max_{u\in\{-1,0,1\}} c_{\theta}(u) = \max\left\{0,\,\,\, \min_{\pi\in\Pi} \mathbb P_{\pi}(Y(1)\ge Y(0)) - \theta,\,\,\, \theta - \max_{\pi\in\Pi} \mathbb P_{\pi}(Y(1)\ge Y(0))\right\}
\]
is the worst of these two one-sided checks, and equals $0$ at $u=0$. If neither side can produce a positive value, then $\theta$ is in the identified set $\Theta_{I}$.

For Gaussian marginals with common variance $\sigma^2$ and means $\mu_1\ge \mu_0$, the classical sharp bounds of \citet{makarov1982estimates} yield
\[
1-2\Phi\left(\frac{-(\mu_1-\mu_0)}{2\sigma}\right) \le \mathbb{P}_\pi\left(Y(1)\ge Y(0)\right) \le\ 1,
\]
as also discussed in \citet{firpo2019partial}. With $\sigma=1$, $\mu_0=0$, and $\mu_1=2$, the identified set is
\(
0.69 \le \mathbb{P}_\pi\left(Y(1)\ge Y(0)\right) \le 1,
\)
which matches our characterization.

Figure \ref{fig:toy-example-population} plots $u\mapsto c_{\theta}(u)$ for $u\in[-1,1]$. The curve is piecewise linear, anchored at
$u=1$ by the lowest feasible beneficiary share minus $\theta$ ($\min_{\pi\in\Pi} \mathbb P_{\pi}(Y(1)\ge Y(0))-\theta=0.69-\theta$), and at $u=-1$ by $\theta$ minus the highest feasible beneficiary share ($\theta-\max_{\pi\in\Pi} \mathbb P_{\pi}(Y(1)\ge Y(0))=\theta-1$). The identification check reduces to verifying whether this curve stays weakly below zero. For $\theta<0.69$, the right endpoint is above zero, whereas for $\theta\in[0.69,1]$, the whole curve remains nonpositive, and hence $\Theta_{I}=[0.69,1]$.

\begin{figure}[t]
  \centering
  \includegraphics[width=0.6\textwidth]{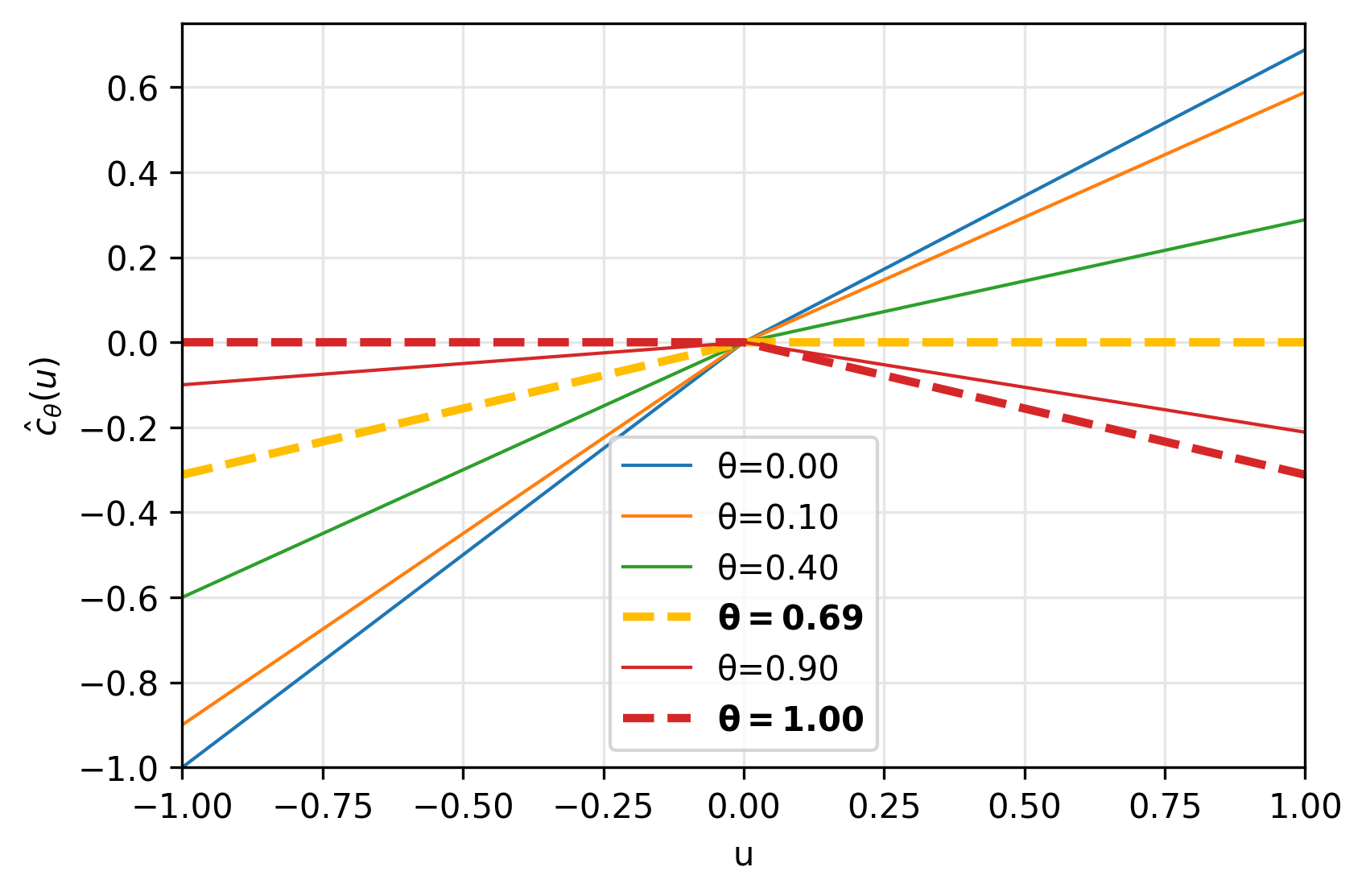}
  \caption{Population cost function $u\mapsto c_{\theta}(u)$ for the RCT toy example}
  \label{fig:toy-example-population}
\end{figure}

\section{Estimation and inference}\label{sec:estimation-and-inference}

Suppose now that we have access to random samples $X_1,\dots, X_n$ and $Y_1,\dots, Y_m$ from distributions $\mu$ and $\nu$, respectively.
We let $n=m$ and assume that the two samples are independent, although these assumptions can be relaxed at the expense of heavier notation.
In this section, we develop an estimator of the sharp identified set $\Theta_I$ and a test of the hypothesis $H_0:\theta=\theta_0$.

\subsection{Estimation}

The characterization in \Cref{thm:id-set-characterization} suggests a natural estimator of the sharp identified set.
Denote by $\hat\mu$, $\hat\nu$ the empirical distributions based on samples $(X_i)$ and $(Y_j)$, and let $\hat\Pi \bydef \Pi(\hat\mu,\hat\nu)$ be the set of joint distributions with marginals $\hat\mu$ and $\hat\nu$.
The sample analog of $c_{\theta}(u)$ is then
\begin{align*}
  \hat c_{\theta}(u) \bydef \min_{\pi\in \hat \Pi} \E_\pi \left[u'\phi(X,Y,\theta)\right].
\end{align*}
This is the value of the empirical OT problem with the cost function $(x,y) \mapsto u'\phi(x,y,\theta)$.
This program is known to be computationally demanding for large samples, sensitive to sampling noise, and its statistical behavior can deteriorate as the data dimension increases \citep[e.g.,][]{cuturi2013sinkhorn,hundrieser2024unifying}.
We therefore suggest using \emph{entropic regularization} -- a classical technique for improving analytical and computational properties of OT that was introduced in \citet{cuturi2013sinkhorn} and a working paper version of \citet{galichon2022cupid}. Entropic regularization adds a penalty for the ``informativeness'' of the transport plan, i.e., the deviation from the product of the empirical marginals $\hat\mu\otimes\hat\nu$:
\begin{align*}
    \hat c_{\theta,\varepsilon}(u) \bydef \min_{\pi\in \hat \Pi} \E_\pi \left[u'\phi(X,Y,\theta)\right]+ \varepsilon \cdot \operatorname{KL}(\pi\,||\,\hat\mu\otimes\hat\nu),
\end{align*}
where $\varepsilon>0$ is a penalty parameter and $\operatorname{KL}$ is the Kullback-Leibler divergence,
\begin{align*}
  \operatorname{KL}(\pi\,||\,\hat\mu\otimes\hat\nu) \bydef \int \log \frac{d\pi}{d(\hat\mu\otimes\hat\nu)}(x,y) \, d\pi(x,y).
\end{align*}
The regularized objective is strictly convex, less sensitive to sampling noise, and regains standard $\sqrt{n}$ asymptotics, as we show in \Cref{subsec:inference}.
Importantly, the regularized program can be solved using explicit iterations of the Sinkhorn algorithm, see \Cref{sec:implementation}.

We suggest estimating $\Theta_I$ by
\begin{align*}
    \widehat\Theta_I \bydef \left\{\theta\in\Theta: \, \hat D_{\varepsilon_n}(\theta) \le \eta_n \right\},
\end{align*}
where $\eta_n>0$ is a tuning parameter and the regularized distance statistic is
\begin{align}
    \hat D_{\varepsilon_n}(\theta) \bydef \max_{u\in \B} \hat c_{\theta,\varepsilon_n}(u). \label{eq:D-sample}
\end{align}
As usual, the computational benefit of regularization comes at the cost of introducing bias.
We show that this bias is asymptotically negligible under appropriate choices of $\varepsilon_n$ and $\eta_n$ in the sense that the estimator $\widehat\Theta_I$ is consistent for $\Theta_I$ in the Hausdorff distance.
To this end, we split the estimation error in $\hat D_{\varepsilon_n}$ into two terms, the sampling error and the regularization bias, as follows:
\begin{align}
  \left|\hat D_{\varepsilon_n}(\theta) - D(\theta)\right| \le \left|\hat D_{\varepsilon_n}(\theta) - D_{\varepsilon_n}(\theta)\right| + \left|D_{\varepsilon_n}(\theta) - D(\theta)\right|. \label{eq:estimation-error-decomposition}
\end{align}
We show in \Cref{thm:eot-stability} below that the first term shrinks at the rate $O_p(n^{-1/d})$ when $d\ge 3$, uniformly over $\varepsilon_n \in [0,1]$, whereas
\Cref{thm:eot-bias} implies that the second term shrinks at the rate $O(\varepsilon_n\log(1/\varepsilon_n))$.
These rates allow us to choose $\varepsilon_n$ and $\eta_n$ to ensure the consistency of $\widehat\Theta_I$.

We maintain the following assumption on the moment function $(x,y) \mapsto \phi(x,y,\theta)$.

\begin{assumption}\label{a:lipschitz-moments}
There exist constants $L_1,L_2<\infty$ such that, for all $\theta\in\Theta$, $x,x'\in\mathcal X$, and $y,y'\in\mathcal Y$,
\begin{align*}
\|\phi(x',y',\theta)-\phi(x,y,\theta)\| &\leq L_1 \|x-x'\| + L_2 \|y-y'\|.
\end{align*}
\end{assumption}

To obtain the convergence rate of the first term in \eqref{eq:estimation-error-decomposition}, we mimic the arguments of \citet{eckstein2022quantitative} and establish that the EOT value $c_{\theta,\varepsilon}(u)$ is Lipschitz continuous with respect to the marginals $\mu,\nu$ in the Wasserstein distance, uniformly in $\varepsilon\geq0$, $u\in\B$, and $\theta\in\Theta$.\footnote{This property is known as \emph{stability} of the EOT value with respect to the marginals.}
Combining this result with known Wasserstein convergence rates of empirical measures \citep[e.g.,][]{fournier2015rate} yields the following result.


\begin{theorem}[Uniform rate for the sampling error] \label{thm:eot-stability}
Suppose that \Cref{a:param-compact,a:supports-compact,a:lipschitz-moments} hold. Then
\begin{equation*}
  \sup_{\varepsilon\geq0}
  \sup_{\theta\in\Theta}
  |\widehat D_\varepsilon(\theta)
  -D_\varepsilon(\theta)|
  = O_p(b_{n,d}),
\end{equation*}
where
\[
b_{n,d}\bydef
\begin{cases}
n^{-1/2},&d=1,\\
n^{-1/2}\log n,&d=2,\\
n^{-1/d},&d\geq3.
\end{cases}
\]
\end{theorem}

\begin{proof}
  See Appendix \ref{app:eot-stability}.
\end{proof}

We now analyze the second term in \eqref{eq:estimation-error-decomposition}, i.e., the bias introduced by regularization, using techniques from the proof of Proposition 3.1 in \citet{carlier2023convergence}.
Although consistency of our estimator only requires the convergence rate of the bias, we provide an explicit bound that we later employ to develop an inference procedure with controlled conservativeness; see \Cref{subsec:inference}.

\begin{theorem}[Explicit uniform bound on the regularization bias]\label{thm:eot-bias}
    Suppose that \Cref{a:param-compact,a:supports-compact,a:lipschitz-moments} hold and the supports $\mathcal X$ and $\mathcal Y$ are contained in axis-aligned cubes of side lengths $B_1$ and $B_2$, respectively.
    Define the function
\begin{align*}
r(\varepsilon) \bydef d\varepsilon\log(1/\varepsilon) +\varepsilon\min_{j=1,2} \left\{L_j+d\log(1+B_j\sqrt d)\right\},
\end{align*}
and set $r(0)=0$.
    Then, for all $\theta\in\Theta$ and $0\leq\varepsilon\leq1$,
  \begin{equation*}
  0\leq D_\varepsilon(\theta)-D(\theta)
  \leq r(\varepsilon).
  \end{equation*}
\end{theorem}

\begin{proof}
  See Appendix \ref{app:eot-bias}.
\end{proof}

The bound satisfies $r(\varepsilon)=O(\varepsilon\log(1/\varepsilon))$ as $\varepsilon\downarrow 0$ and depends only on the underlying dimension, support bounds, and Lipschitz constants.

\begin{remark}\label{rem:eot-bias-mixed}
Suppose that the data have $d_c \le d$ continuous variables. Then the rate $b_{n,d}$ in \Cref{thm:eot-stability} can be replaced with a weakly faster rate $b_{n,d_c}$, and the bound $r(\varepsilon)$ in \Cref{thm:eot-bias} can be replaced with a weakly tighter bound
\begin{align*}
  \tilde r(\varepsilon) \bydef \min\{r(\varepsilon),\,r_{\text{mix}}(\varepsilon)\},
\end{align*}
where
\begin{align*}
r_{\text{mix}}(\varepsilon) \bydef d_c\varepsilon\log(1/\varepsilon) +\varepsilon\min_{j=1,2} \left\{L_{j,c}+\log K_j+d_c\log(1+B_{j,c}\sqrt{d_c})\right\}.
\end{align*}
Here $X=(X_c,X_d)$ and $Y=(Y_c,Y_d)$, where $X_c$ and $Y_c$ are continuous variables of dimension $d_c$, and $X_d$ and $Y_d$ are discrete variables with $K_1$ and $K_2$ possible values, respectively.
We denote by $B_{1,c}$ and $B_{2,c}$ the side lengths of axis-aligned cubes containing the supports of $X_c$ and $Y_c$ and denote by $L_{1,c}$ and $L_{2,c}$ the Lipschitz constants of $\phi$ for continuous deviations, i.e.,
\begin{align*}
\|\phi((x_c',x_d),y,\theta)-\phi((x_c,x_d),y,\theta)\|
&\leq L_{1,c}\|x_c'-x_c\|,\\
\|\phi(x,(y_c',y_d),\theta)-\phi(x,(y_c,y_d),\theta)\|
&\leq L_{2,c}\|y_c'-y_c\|,
\end{align*}
for all $\theta\in\Theta$ and all possible discrete values $x_d$, $y_d$.
Notice that $B_{j,c}$ and $L_{j,c}$ can be chosen to satisfy $B_{j,c}\le B_j$ and $L_{j,c}\le L_j$.
\end{remark}

In view of \Cref{thm:eot-stability,thm:eot-bias}, we need $\varepsilon_n$ and $\eta_n$ to satisfy the following assumption on their relative rates of convergence.

\begin{assumption}\label{a:consistency-tuning-parameters}
    The tuning parameters $\varepsilon_n$ and $\eta_n$ satisfy $\varepsilon_n \downarrow 0$, $\eta_n \downarrow 0$, and
    \[
    b_{n,d} + \varepsilon_n \log(1/\varepsilon_n) = o(\eta_n).
    \]
\end{assumption}

Under this choice of tuning parameters, $\widehat\Theta_I$ is consistent for $\Theta_I$ in the Hausdorff distance. Moreover, $\widehat\Theta_I$ converges to $\Theta_I$ ``from the outside'' with probability approaching one.

\begin{theorem}[Consistency]\label{thm:consistency}
Suppose that \cref{a:id-set-characterization,a:lipschitz-moments,a:consistency-tuning-parameters} hold. Then $\Pr(\Theta_I\subseteq\widehat\Theta_I)\to1$ and 
\[
d_H(\widehat\Theta_I,\Theta_I)\xrightarrow{p}0.
\]
\end{theorem}

\begin{proof}
  See \Cref{app:consistency}.
\end{proof}

\subsection{Inference}\label{subsec:inference}

In this section, we develop a test of the hypothesis $H_0:\theta_0 \in \Theta_I$, where $\theta_0 \in \Theta$ is the hypothesized value.
A confidence region for the parameter of interest can then be constructed by inverting this test.
Since we still want to enjoy computational and analytical benefits of entropic regularization, we consider the relaxed null hypothesis $H_0:\theta_0 \in \Theta_{I,\varepsilon}$, where
\begin{equation*}
\Theta_{I,\varepsilon}
\bydef 
\left\{
\theta\in\Theta:
D_\varepsilon(\theta)\leq r(\varepsilon)
\right\}
\end{equation*}
and $r(\varepsilon)$ is the explicit bound from \Cref{thm:eot-bias}.
The regularization parameter $\varepsilon$ can be chosen to bound the enlargement $\Theta_{I,\varepsilon}$ of the identified set $\Theta_I$.
We formalize this in the following theorem.

\begin{theorem}\label{thm:regularized-identified-set}
Let $0<\varepsilon\le 1$ and suppose that the assumptions of \Cref{thm:eot-bias} hold. Then $\Theta_{I}\subseteq\Theta_{I,\varepsilon}$. Moreover, if there exist
constants $\kappa>0$ and $\alpha>0$ such that
\begin{equation}
D(\theta)
\geq
\kappa\,d(\theta,\Theta_{I})^\alpha
\qquad\text{for all }\theta\in\Theta,
\label{eq:identification}
\end{equation}
then
\begin{equation*}
d_H(\Theta_{I,\varepsilon},\Theta_{I})
\leq
\left(\frac{r(\varepsilon)}{\kappa}
\right)^{1/\alpha}.
\end{equation*}
\end{theorem}

\begin{proof}
See Appendix \ref{app:regularized-identified-set}.
\end{proof}

\Cref{thm:regularized-identified-set} suggests that we can use the rescaled distance $\sqrt{n}(\hat D_\varepsilon(\theta_0)-r(\varepsilon))$ as the test statistic.
We characterize its asymptotic behavior by first establishing a novel CLT for the regularized OT value $\hat c_{\theta_0,\varepsilon}(u)$, uniformly in $u\in \B$ and $\theta_0\in\Theta$, and then applying the functional delta method to obtain the limiting distribution of the test statistic.
Since this distribution does not have a simple form, we establish the validity of the bootstrap for directionally differentiable functionals of \citet{fang2019inference} using a result in \citet{franguridi2025set}.

\begin{remark}[KS vs.\ CvM statistics]\label{rem:ks-cvm}
  Our (population) statistic
  \[
    D_\varepsilon(\theta)=\max_{u\in \B} c_{\theta,\varepsilon}(u)
  \]
  uses maximization that is characteristic of Kolmogorov-Smirnov (KS) type statistics in the moment inequality literature, see, e.g., \citet{andrews2013inference}.
  The KS-type statistics are powerful against local alternatives with few violations \citep{armstrong2015asymptotically,armstrong2018choice}, and provide diagnostics for the most binding inequality. 
  
  Alternatively, we could consider the Cram\'er-von Mises (CvM) type statistic
  \[
  \tilde D_\varepsilon(\theta) = \int_{\B} c_{\theta,\varepsilon}(u)_{+}^{2} \, d \omega(u),
  \]
  for a probability measure $\omega$ on $\B$ with full support.
  The CvM-type statistics are powerful against diffused local alternatives \citep{andrews2013inference}, retain power under weak identification \citep{bugni2010bootstrap}, typically exhibit smaller finite-sample size distortions, and are less sensitive to slack inequalities.
  While our methodology can be extended to the CvM type statistic, we focus on the KS type statistic in this paper since it is simple to implement and demonstrates good size and power performance in the Monte Carlo simulations. 

  Finally, yet another choice of a test statistic is a self-normalized moment violation statistic as in \citet{chetverikov2018adaptive}. We leave consideration of such a statistic for future work.
  \end{remark}

To establish our uniform CLT, we impose the following assumption.

\begin{assumption}\label{a:UCLT}
  For some integer \(s>d/2\), each moment function \(\phi_j(\cdot,\cdot,\theta)\) belongs to the H\"older space \(C^s(\mathcal X\times\mathcal Y)\) for every \(\theta\in\Theta\), and
  \[
  \max_{1\le j\le p}\sup_{\theta\in\Theta}
  \|\phi_j(\cdot,\cdot,\theta)\|_{C^s(\mathcal X\times\mathcal Y)}
  <\infty.
  \]
\end{assumption}
\Cref{a:UCLT} imposes uniform smoothness of the moment functions $\phi_j$ on the entire set $\X\times\Y$ for simplicity. However, when some components of $X$ and $Y$ are (finitely) discrete and others are continuous, this assumption can be extended to allow the functions $\phi_j$ to be smooth only with respect to the continuous components.
For example, in the panel logit model, the moment function includes the indicator $\mathbf{1}\{Y_1+Y_2=1\}$ for binary $Y_t\in\{0,1\}$, which is covered by this extension.
We also note that the maintained assumption of random sampling can be relaxed to accommodate dependence within samples, e.g., when the samples are stationary $\beta$-mixing processes, as well as dependence across samples, see Remark 11 in \citet{goldfeld2024statistical}.
  
\begin{theorem}[Uniform CLT for regularized OT value]\label{thm:UCLT}
  Suppose that \Cref{a:param-compact,a:supports-compact,a:UCLT} hold.
  Then there exists a tight Gaussian process $\G_\varepsilon$ on $C(\B \times \Theta)$ such that
    \begin{align*}
      \sqrt{n}\left( \hat c_{\theta,\varepsilon}(u)-c_{\theta,\varepsilon}(u) \right) \weakto \G_\varepsilon(u,\theta) \text{ in } C(\B\times\Theta).
    \end{align*}  
  \end{theorem}

  \begin{proof}
    See \Cref{app:UCLT}.
  \end{proof}

  \begin{remark}
    If $\dim\phi=1$, then considering this statement at a single point $u=1$ and given $\theta$ delivers asymptotic normality of the regularized OT value under a general smooth cost function $\phi$. To our knowledge, this is the first such result available in the literature. In our derivation, however, we relied heavily on the arguments in \citet{mena2019statistical}, who were the first to establish the asymptotic normality when the cost function is quadratic, and \citet{goldfeld2024statistical}, who extended this result using empirical processes theory.
  \end{remark}

  \begin{remark}
    \Cref{a:supports-compact} is needed to establish uniform bounds on the optimal potentials and their derivatives. At the expense of more complicated proofs, this assumption can be relaxed to bounds on the tails of $\mu,\nu$ related to the cost function, similar to \citet{goldfeld2024statistical}.
  \end{remark}

Applying the functional delta method to this uniform convergence statement with the functional $\chi(c) \bydef\max_{u\in \B} c(u)$ leads to the following result.

\begin{corollary}[asymptotic distribution of the distance statistic]\label{cor:D-distribution}
  Suppose that \Cref{a:UCLT} holds.
  Then, for any $\theta_0 \in \Theta$,
    \begin{align*}
        \sqrt{n}(\hat D_\varepsilon(\theta_0) - D_\varepsilon(\theta_0)) = \sqrt n\left(\chi(\hat c_{\theta_0,\varepsilon})-\chi(c_{\theta_0,\varepsilon})\right)\weakto \max_{u\in U_{c,\varepsilon}(\theta_0)}\G_\varepsilon(u,\theta_0),
      \end{align*}
      where $U_{c,\varepsilon}(\theta_0) \bydef \arg\max_{u\in\B} c_{\theta_0,\varepsilon}(u)$.
\end{corollary}

\begin{proof}
  See \Cref{app:D-distribution}.
\end{proof}

The asymptotic distribution of the distance statistic is neither available in closed form nor easy to simulate from because it depends on the unknown features of the data-generating process such as the argmax set $U_{c,\varepsilon}(\theta_0)$.
Moreover, standard bootstrap often fails to control size uniformly in partially identified models \citep{andrews2009validity,andrews2009invalidity}.
This failure can occur when the parameter is on the boundary of the parameter space \citep{andrews2000inconsistency}.
In our model, the relevant directional derivative is $h\mapsto\max_{u\in U_{c,\varepsilon}(\theta_0)}h(u)$, which need not be linear when $U_{c,\varepsilon}(\theta_0)$ is not a singleton.

\begin{algorithm}[!t]
  \caption{Testing $H_0:\theta_0 \in \Theta_{I}$}
  \label{alg:bootstrap-test}
  \begin{algorithmic}[1]
  \Require $\theta_0$, $\varepsilon$, $\iota_n$, number of bootstrap draws $B$
  \Ensure Test decision for $H_0$

  \smallskip\Statex \textit{EOT estimates:}
  \State Compute $\hat c_{\theta_0,\varepsilon}(u)$ using the Sinkhorn algorithm with regularization parameter $\varepsilon$
  \State Compute $\hat D_\varepsilon(\theta_0) \leftarrow \max_{u\in\mathbb B} \hat c_{\theta_0,\varepsilon}(u)$
  \State Define the set
      \[
          \hat U_{n,\varepsilon} \leftarrow \left\{u\in\mathbb B :\, \hat c_{\theta_0,\varepsilon}(u) \ge \max_{\|v\|\le 1} \hat c_{\theta_0,\varepsilon}(v) - \iota_n \right\}
      \]

  \smallskip\Statex \textit{Bootstrap:}
  \For{$b = 1,\dots,B$}
      \State Draw bootstrap samples $X_1^b,\dots,X_n^b$ and $Y_1^b,\dots,Y_n^b$ with replacement from $\hat\mu$ and $\hat\nu$
      \State Compute $\hat c_{\theta_0,\varepsilon}^{*,b}(u)$, the regularized OT value on \{$X_i^b,Y_i^b\}_{i=1}^n$, using the Sinkhorn algorithm
      \State Compute the bootstrap statistic
      \[
          \hat T_\varepsilon^{*,b} \gets
          \max_{u\in \hat U_{n,\varepsilon}}
          \sqrt{n}\left(\hat c_{\theta_0,\varepsilon}^{*,b}(u) - \hat c_{\theta_0,\varepsilon}(u)\right)
      \]
  \EndFor

  \smallskip\Statex \textit{Critical value and test decision:}
  \State Let $\hat q_{\alpha,\varepsilon}^*$ be the $(1-\alpha)$ empirical quantile of
      $\{\hat T_\varepsilon^{*,1}, \dots, \hat T_\varepsilon^{*,B}\}$
  \State \Return Reject $H_0$ if $\sqrt{n}(\hat D_\varepsilon(\theta_0)-r(\varepsilon)) > \hat q_{\alpha,\varepsilon}^*$

  \end{algorithmic}
\end{algorithm}

To overcome this challenge, we make use of the bootstrap for directionally differentiable functionals of \citet{fang2019inference}.
Our testing procedure is described in \Cref{alg:bootstrap-test} and depends on an additional tuning parameter $\iota_n$.
We impose the following assumption.

\begin{assumption}[bootstrap validity] \label{a:bootstrap-validity}
  \begin{subassumption}
    \item \label{a:sharp-maxima} There exists a nondecreasing function $\rho:\R_+\to\R_+$ with $\rho(0)=0$ and $\rho(\delta)>0$ for all $\delta>0$ such that
    \begin{align*}
        c_{\theta_0,\varepsilon}(u) \le \max_{v\in \B}c_{\theta_0,\varepsilon}(v) - \rho\left(d(u,U_{c,\varepsilon}(\theta_0))\right) \text{ for all } u \in \B,
    \end{align*}
    where $d(u,U)=\inf_{v\in U}\|u-v\|$.
    \item \label{a:bandwidth}
    $\iota_n \downarrow 0$ and $\sqrt n\iota_n\to\infty$.
  \end{subassumption}
\end{assumption}

\Cref{a:sharp-maxima} is a weak version of the standard linear separation condition that accommodates our regularized criterion. In particular, $\rho(\delta)=\kappa\delta^2$ allows quadratic decay near an interior maximizer, where linear separation fails.
It also suffices for this assumption to hold in a small neighborhood around $U_{c,\varepsilon}(\theta_0)$ rather than on the entire ball $\B$, since continuity on the compact ball ensures a positive gap away from that neighborhood.
\Cref{a:bandwidth} requires $\iota_n$ to converge to zero slower than $n^{-1/2}$.
This guarantees that the enlarged argmax $\hat U_{n,\varepsilon}$ converges in the Hausdorff distance to the true argmax $U_{c,\varepsilon}(\theta_0)$.
If $U_{c,\varepsilon}(\theta_0)$ is known to be a singleton, we can take $\hat U_{n,\varepsilon}$ to be the sample argmax of $\hat c_{\theta_0,\varepsilon}$.

It is straightforward to establish that under \Cref{a:bootstrap-validity}, our testing procedure controls size locally uniformly in the sense of Corollary 3.2 of \citet{fang2019inference}.
We refer the reader to Theorem 4 of \citet{franguridi2025set} for details.

\begin{remark}
  Applied researchers are often interested in one component (or a given function) of the parameter vector $\theta$. In the context of partially identified models, this problem is called subvector inference, and our inference procedure would be adaptable to this case as well.
  To see that, let $\theta=(\theta_1,\theta_2')$, where the scalar $\theta_1$ is the parameter of interest, and let the parameter space be $\Theta=\Theta_1\times \Theta_2$.
  Notice that the sharp identified set for $\theta_1$ is 
  \begin{align*}
    \Theta_{I}^{(1)} = \{\theta_1\in\Theta_1: \,\, D^{(1)}(\theta_1)=0\},
  \end{align*}
  where $D^{(1)}(\theta_1) = \min_{\theta_2 \in \Theta_2} D(\theta_1,\theta_2)$.
The regularized profile $D_{\varepsilon}^{(1)}(\theta_1)=\min_{\theta_2\in\Theta_2}D_\varepsilon(\theta_1,\theta_2)$ can then be estimated by $\hat D_{\varepsilon}^{(1)}(\theta_1)=\min_{\theta_2\in\Theta_2} \hat D_\varepsilon(\theta_1,\theta_2)$, and its uniform weak limit can be established by the delta method. Finally, the \citet{fang2019inference} bootstrap applies again because the functional $h \mapsto \min_{\theta_2 \in \Theta_2} h(\theta_1,\theta_2)$ is Hadamard directionally differentiable.
\end{remark}

\begin{remark}
  Taking the minimum of $\sqrt{n} (\hat D_{\varepsilon}(\theta)-r(\varepsilon))$ over the parameter space $\theta\in\Theta$ yields a statistic that can be used to develop a specification test, i.e., a test of the hypothesis $\Theta_I \neq \varnothing$, see, e.g., \citet{bugni2015specification} for a similar idea in the context of moment inequality models.
  The bootstrap of \citet{fang2019inference} can then be employed to obtain critical values for such a test.
\end{remark}

\subsection{Numerical implementation}\label{sec:implementation}

Our procedure requires computing the distance statistic \eqref{eq:D-sample}, which amounts to solving two nested optimization problems.

The inner problem (entropic OT) is solvable by the Sinkhorn algorithm described in \Cref{alg:sinkhorn}.
The outer problem is the maximization of a concave function $\hat c_{\theta,\varepsilon}(u)$, which is the output of \Cref{alg:sinkhorn}, over the unit ball $u\in\B$.
Many other off-the-shelf algorithms are available for constrained concave optimization. See \citet{bubeck2015convex} for a comprehensive review.

\begin{algorithm}[tp]
  \caption{Sinkhorn algorithm for entropic OT with cost $(x,y) \mapsto u'\phi(x,y,\theta)$}
  \label{alg:sinkhorn}
  \begin{algorithmic}[1]
  \Require Samples $(x_i)_{i=1}^n$, $(y_j)_{j=1}^m$, moment function $\phi(\cdot,\cdot,\theta)$, $u\in\B$, $\varepsilon>0$, tolerance $\tau>0$
  \Ensure Approximate value $\hat c_{\theta,\varepsilon}(u)$
  
  \smallskip\Statex \textit{Construct cost matrix:}
  \State $C_{ij} \leftarrow u'\phi(x_i,y_j,\theta)$ \quad for all $i=1,\dots,n$, $j=1,\dots,m$
  
  \smallskip\Statex \textit{Initialize dual potentials:}
  \State $f\leftarrow0_n$, $g\leftarrow0_m$
  
  \smallskip\Statex \textit{Iterate to enforce marginal constraints:}
  \Repeat
      \State $f_i^{\rm new}\leftarrow-\varepsilon\left[\log\left(\sum_{j=1}^m\exp\left(\frac{g_j-C_{ij}}{\varepsilon}\right)\right)-\log m\right]$ for all $i$
      \State $e\leftarrow\frac{1}{n}\sum_{i=1}^n\left|\exp\left(\frac{f_i-f_i^{\rm new}}{\varepsilon}\right)-1\right|$
      \State $f\leftarrow f^{\rm new}$
      \State $g_j\leftarrow-\varepsilon\left[\log\left(\sum_{i=1}^n\exp\left(\frac{f_i-C_{ij}}{\varepsilon}\right)\right)-\log n\right]$ for all $j$
  \Until{$e<\tau$}
  
  \smallskip\Statex \textit{Compute entropic OT value:}
  \State $\hat c_{\theta,\varepsilon}(u)\leftarrow\frac{1}{n}\sum_{i=1}^n f_i+\frac{1}{m}\sum_{j=1}^m g_j$
  
  \smallskip\State \Return $\hat c_{\theta,\varepsilon}(u)$
  \end{algorithmic}
\end{algorithm}

Note that, after each update of $g$, the implied coupling $P_{ij}=(nm)^{-1}\exp((f_i+g_j-C_{ij})/\varepsilon)$ has the correct column sums and total mass one. The dual value is therefore $n^{-1}\sum_i f_i+m^{-1}\sum_j g_j$. After updating $f$, we calculate its absolute deviation from the previous iterate and use it as the stopping criterion. The algorithm does not require computing or storing the coupling.

Finally, computing the estimator of the identified set and the associated confidence region requires setting the tuning parameters $\varepsilon_n$, $\eta_n$, and $\iota_n$.
To satisfy \Cref{a:consistency-tuning-parameters,a:bandwidth}, these parameters can be chosen as
\[
\varepsilon_n=\frac{c_\varepsilon b_{n,d_c}}{\log n},\quad
\eta_n=c_\eta b_{n,d_c}\sqrt{\log n},\quad
\iota_n=c_\iota\sqrt{\frac{\log n}{n}},
\]
where $c_\varepsilon$, $c_\eta$, $c_\iota>0$ are constants and $d_c$ is the number of continuous variables, see \Cref{rem:eot-bias-mixed}.

\section{Illustrative examples}\label{sec:examples}

This section presents three empirical examples that illustrate how our methods can address partial identification problems arising from missing or incomplete data. The examples span different areas of econometrics: panel data, instrumental variables, and macro-finance. In each case, we show how the inability to observe certain joint distributions leads to partial identification of parameters of interest and how our methodology can be applied to characterize the identified sets in various contexts.

\subsection{Fixed effects panel logit with attrition and refreshment}\label{sec:logit}

This is our leading example and forms the basis for the Monte Carlo simulations in \Cref{sec:simulations}. Panel logit models are widely used in empirical studies to analyze binary outcomes while controlling for individual heterogeneity through fixed effects. However, attrition is a common issue in panel data, where units may drop out of the sample in subsequent periods for reasons potentially correlated with outcomes of interest. When a refreshment sample is available, the model can be partially identified via OT. With an adjustment that exploits the panel structure, our methodology can be used to estimate and construct confidence regions for the identified set for the common slope parameter and the AME. 

\subsubsection{Common slope parameter}
The common slope parameter is point identified under complete data or when attrition is independent of outcomes conditional on observables, but becomes partially identified under unrestricted attrition.
For notational simplicity, consider a static panel logit model with two periods $T=2$,
\begin{align}
Y_{it} = \mathbf{1}\left\{X_{it}'\theta + \alpha_i - \varepsilon_{it} > 0\right\}, \label{eq:logit-model}
\end{align}
where $\theta$ captures the effect of covariates on the outcome, $\alpha_i$ represents individual fixed effects, and $\varepsilon_{it}$ follows a standard logistic distribution. Let $Z_{it}=(Y_{it},X_{it})$. 

The standard approach to eliminate the incidental parameter $\alpha_i$ is to condition on the sufficient statistic $S_i = Y_{i1} + Y_{i2}$. For individuals with $S_i = 1$ (switchers), the conditional log-likelihood is
\begin{align*}
\ell_i(\theta \mid S_i=1) 
= Y_{i1}X_{i1}'\theta + Y_{i2}X_{i2}'\theta - \ln\left(e^{X_{i1}'\theta}+e^{X_{i2}'\theta}\right),
\end{align*}
with the corresponding conditional score function $s(Z_{i1},Z_{i2},\theta)$ and the moment condition
\begin{align*}
\E\left[s(Z_{i1},Z_{i2},\theta_0) \mid S_i=1\right] = 0.
\end{align*}
The explicit form of the score is given in \Cref{app:logit-beta}. To translate this moment condition into our OT framework, we embed the event $\{Y_{i1}+Y_{i2}=1\}$ into the cost function and define
\begin{align*}
\phi\left(z_1,z_2,\theta\right) 
= s(z_1,z_2,\theta)\mathbf{1}\left\{y_1+y_2=1\right\}.
\end{align*}
Since the outcome variables are binary, the indicator $\mathbf{1}\left\{Y_1+Y_2=1\right\}$ simply indexes a finite collection of cost functions in the continuous covariates, and is therefore covered by the finite discrete component extension discussed after \Cref{a:UCLT}.

A naive approach to partially identify $\theta$ would use only the marginal distributions from period 1 (original sample) and period 2 (refreshment sample). However, we can achieve tighter bounds by exploiting the panel structure. Since we observe both periods for retainers, we fix their joint distribution $f_{1,2\mid\text{ret}}$ and apply OT only to couple the attriter distributions $f_{1\mid\text{att}}$ and $f_{2\mid\text{att}}$ across periods. While $f_{1\mid\text{att}}$ is directly observed, $f_{2\mid\text{att}}$ is unobserved because attriters are missing in period 2, but it can be recovered from the observed refreshment and retainer samples via the law of total probability. See \Cref{app:logit-beta} for details.

Let $p$ denote the retention rate. Since the joint distribution of retainers is observed, the directional OT criterion is
\begin{align*}
c_{\theta}(u)
&= p \cdot \E_{f_{1,2\mid\text{ret}}}\left[u'\phi(Z_1,Z_2,\theta)\right]
+ (1-p) \cdot \inf_{f \in \Pi(f_{1\mid\text{att}}, f_{2\mid\text{att}})} \E_f\left[u'\phi(Z_1,Z_2,\theta)\right], \quad u\in\B.
\end{align*}
The distance $D(\theta)$ and identified set $\Theta_{I}$ are defined as before.
\Cref{alg:fe-logit-bounds} in \Cref{app:logit-beta} presents our estimation procedure.

Our estimator naturally extends to dynamic panel logit models where lagged dependent variables appear as regressors. It is possible to incorporate the moment conditions developed by \citet{honore2024moment} as the cost function, with a similar but more complicated partition structure separating retainers and attriters to achieve efficiency gains.

\subsubsection{AME}

We now consider estimation of the AME, which captures the average change in outcome probability induced by a marginal change in a covariate and provides a directly interpretable measure for empirical work. Specifically, the AME of covariate $j$ at period $\tau$ is defined as
\[
\delta_{\tau,j} \bydef \theta_j\E[\Lambda(X_{\tau}'\theta + \alpha)(1-\Lambda(X_{\tau}'\theta + \alpha))].
\]
The AME is partially identified even without attrition due to the incidental parameters problem combined with the nonlinear structure of the logit model. Under unrestricted attrition, this identification issue becomes more severe as the joint distribution of outcomes across periods is no longer observable.

\citet{davezies2021identification} show how to construct outer bounds on the AME without attrition. They show that $\delta_{\tau,j}$ belongs to the interval $\tilde{\delta} \pm \bar{b}$, where $\tilde{\delta} = \E[p(X_{1:T},S,\theta_0)]$ and $\bar{b} = \E[a(X_{1:T},S,\theta_0)]$ for functions $p$ and $a$ constructed from a degree-$(T+1)$ Chebyshev polynomial. The explicit formulas for $p$, $a$, and the associated coefficients are collected in \Cref{app:logit-ame-no-attrition}, including closed-form expressions for the case $T=2$.

Under unrestricted attrition, we extend the partition approach for $\theta$ to estimate the identified set for the AME.
First, we compute the identified set $\widehat\Theta_I$ for the common slope parameters $\theta$ as in \Cref{alg:fe-logit-bounds} and construct a finite grid $\left\{\theta^{(g)}\right\}\subset \widehat\Theta_I$.
Second, for each grid point $\theta^{(g)}$ we plug it into the Chebyshev approximation and, using our partition of retainers and attriters, compute the corresponding identified set for the AME conditional on $\theta^{(g)}$, denoted $\left[\underline{\delta}(\theta^{(g)}), \overline{\delta}(\theta^{(g)})\right]$, via OT.
Finally, we profile over $\theta$ by taking the union to obtain the identified set for the AME:
\(
\bigcup_g \left[\underline{\delta}(\theta^{(g)}), \overline{\delta}(\theta^{(g)})\right].
\)
See \Cref{app:logit-ame-attrition} for details.

For inference, we can replace the estimated identified set $\widehat\Theta_I$ in Step 1 with the confidence region for $\theta$ obtained by inverting the test in \Cref{alg:bootstrap-test}, and then profile over this confidence region to obtain a confidence region for the AME.

Although this grid-based profiling yields conservative estimates of the AME identified set due to the Chebyshev polynomial approximation and the two-stage approach that first estimates the identified set for $\theta$ and then the identified set for $\delta$, it delivers a significant improvement in computational efficiency over alternative methods such as Hankel moment matrix positivity or a single-stage approach that estimates both $\theta$ and $\delta$ simultaneously.

\subsection{Nonparametric IV without large support}

IV methods are commonly used in causal inference when strong ignorability fails. The classical control function approach to IV requires a large support condition for point identification. However, this assumption often fails in practice when the treatment is discrete or has limited variation, such as binary treatments in \citet{angrist1996identification} and discrete treatment intensity in \citet{angrist1995two}.
This example shows how our methodology can deliver meaningful bounds even when the large support assumption fails.

Consider the standard nonparametric IV model
\[
  Y = g(X,V), 
  \quad
  X = h(Z,W),
  \quad
  (W,V)\perp Z,
\]
where $Y$ is the outcome, $X$ is the endogenous treatment, $Z$ is the instrument, and $(V,W)$ are the unobserved errors.
We want to estimate a generic object of interest \(\theta=\E[\Lambda(g(X,V))]\) beyond the local average treatment effects (LATE). For example, setting $\Lambda(y)=y$ gives the average potential outcome $\E[g(x,V)]$ at treatment level $x$, and hence the average treatment effect (ATE) between two treatment levels $x_1$ and $x_0$ is $\E[g(x_1,V)] - \E[g(x_0,V)]$.
The classical control variable formula requires not only treatment monotonicity ($w \mapsto h(z,w)$ strictly increasing) but also large support ($\supp(R\mid V=v) = \supp(R)$ for all $v$, where $R = F_{X\mid Z}(X)$), see, e.g., Section 4.1 in \citet{gunsilius2025primer}. When this large support condition fails, $\theta$ becomes partially identified.

To formalize this setting, we impose two assumptions. First, similar to the classical setup, we assume treatment monotonicity: $w \mapsto h(z,w)$ is strictly increasing for each $z$, so under a monotone transformation, we can define the control variable $W=F_{X\mid Z}(X)$. Second, we also assume outcome monotonicity: $v \mapsto g(x,v)$ is strictly increasing, so define $V = F_{Y\mid X}(Y) \sim \text{Uniform}[0,1]$. Note that $W$ and $V$ are, in general, not independent. If $W$ has limited variation, such as discrete or censored treatment, the large support condition may fail, i.e., $\supp(W\mid V=v) \subsetneq \supp(W)$ for some $v$.

In many empirical applications, the instrument $Z$ is supported on a finite number of values. For example, \citet{angrist1995two} use quarter-of-birth dummies as instruments for schooling, and \citet{card1995geographic} uses a binary instrument based on proximity to four-year colleges.
For each $z$, the marginals $F_{W}$ and $F_V$ can be recovered from the data, and the moment function is
\[
  \phi(w,v;z) = \Lambda\left(g\left(F_{X\mid Z=z}^{-1}(w), v\right)\right).
\]
Then the sharp identified set for $\theta$ is the interval $\left[\underline\theta, \overline\theta\right]$ with
\[
  \underline\theta
  = \inf_{F\in\Pi(F_{W}, F_V)}
    \sum_z\E_F [\phi(W,V;z)]\Pr(Z=z),\quad
    \bar\theta
    = \sup_{F\in\Pi(F_{W}, F_V)}
      \sum_z\E_F [\phi(W,V;z)]\Pr(Z=z).
\]

\subsection{Euler equation estimation with repeated cross-sections}

A prominent example in macro-finance is estimating the discount factor $\beta$ and risk aversion $\gamma$ from the constant relative risk aversion (CRRA) Euler equation
\[
  \E\left[\beta\left(C_{i,t+1}/C_{it}\right)^{-\gamma}R_{t+1} - 1\mid \mathcal I_{it}\right] = 0,
\]
where $C_{it}$ is individual $i$'s consumption at time $t$, $R_{t+1}$ is the common asset return, and $\mathcal I_{it}$ is the information set.

In practice, this single nonlinear conditional moment is converted into an overidentified unconditional GMM by introducing a $k$-dimensional vector of instruments with $k\ge2$,
\(
  Z_{it} = \left(Z_{t}^A, Z_{it}^I\right),
\)
where $Z_t^A$ are lagged macro variables (such as GDP growth rates and interest rates), and $Z_{it}^I$ are lagged individual variables (such as demographics, as well as prior income and consumption). The moment condition then becomes
\[
  \E\left[Z_{it}\left(\beta\left(C_{i,t+1}/C_{it}\right)^{-\gamma}R_{t+1} - 1\right)\right] = 0.
\]
Under a standard rank condition, this equation can be used to estimate $(\beta,\gamma)$.

To estimate the Euler equation, one would ideally use data that preserve cross-sectional heterogeneity while providing a sufficient sample size, which in practice motivates the use of repeated cross-sections or short rotating panels. Much of the early empirical literature, however, relied on aggregate time-series consumption and return data, as in \citet{hansen_singleton1982}. Because the Euler equation is nonlinear in consumption, Jensen's inequality implies that the nonlinear moment evaluated at aggregate consumption would differ from the cross-sectional average of individual-level terms. This mismatch can induce systematic bias in GMM estimates. Subsequent work emphasized granular data to retain heterogeneity. For example, \citet{dynan_skinner_zeldes2004} exploit the panel structure of the Panel Study of Income Dynamics (PSID), but the PSID has a relatively small cross-sectional sample size. In contrast, many large-scale household surveys, such as the Consumer Expenditure Survey, offer rich cross-sectional coverage via repeated cross-sections or short rotating panels, making them well-suited for our OT-based approach. See also \citet{liu2023full} on estimating full structural models with repeated cross-sections.

In an extreme case, suppose we observe only repeated cross-sections, so that the marginal distributions
\(
  f_{C_{it},Z_{it}^I}(c,z^I)
  \text{ and }
  f_{C_{i,t+1}}(\tilde c),
\)
are known, in contrast to their joint law. Let $\theta = (\beta,\gamma)'$. Then the parameter is only partially identified with the identified set
\(
  \Theta
  = \left\{\theta :\, 
     \underline M(\theta)\le0\le \overline M(\theta)\right\},
\)
where
\[
  \underline M(\theta)
  = \inf_{f\in\Pi\left(f_{C_{it},Z_{it}^I}, f_{C_{i,t+1}}\right)}
    \E_f\left[\phi(Z_{it},C_{it},C_{i,t+1},R_{t+1};\theta)\right],
\]
and $\overline M(\theta)$ is the supremum of the same expression. Here $\Pi(\cdot)$ is the set of all couplings consistent with the observed marginals, and the moment function is 
\[
  \phi(z,c,\tilde c,r;\theta)
  = z'\left(\beta\left(\tilde c/c\right)^{-\gamma}r - 1\right).
\]

With short rotating panels, where households are observed for only several periods, the OT problem becomes more complex: one can exploit the limited longitudinal links to tighten the identified set while using OT for the remaining unlinked portions of the data.

\section{Monte Carlo simulation}\label{sec:simulations}

In this section, we illustrate the performance of our procedure in a Monte Carlo simulation.

As a data-generating process, we use the fixed effects panel logit with attrition and refreshment as described in \Cref{sec:logit} and defined in \eqref{eq:logit-model}. The true parameter is set as $\theta_0=(1.0,2.0)'$.
The covariate vector $X_{it}=(X_{it,1},X_{it,2})'$ consists of two independent components. The first component is uniformly distributed on $[-0.75,0.75]$, and the second has a discrete uniform distribution on $\mathcal{X}_2 = \{-0.75,0,0.75\}$. Both components are independent across units $i$ and time $t$.
The fixed effects $\alpha_i$ are drawn i.i.d.\ from the standard normal distribution.
The idiosyncratic error terms $\varepsilon_{it}$ are drawn i.i.d.\ from the standard logistic distribution with zero mean and unit scale.
The sizes of both the first-period and the refreshment samples are $n_\text{org}=n_\text{ref}=15000$.
Finally, the attrition rate is $1-p=10\%$, and the units drop out of the sample completely at random.

We conduct $500$ simulations. For each simulation, we construct the 90\% confidence region by inverting the test in \Cref{alg:bootstrap-test} on a $40\times40$ grid of hypothesized values $\theta^* \in [-0.25,2.25]\times [0.75,3.25]$ centered at the true value $\theta_0=(1.0,2.0)'$.
We estimate the attriter marginals on the observed covariate values and solve the resulting empirical entropic OT problem using \Cref{alg:sinkhorn}. Since maximizing the concave criterion over the unit ball is a convex optimization problem, we calculate the distance $\hat D_\varepsilon(\theta^*)$ efficiently using sequential quadratic programming. We apply the bias correction in \Cref{app:logit-bias} when inverting the test.
We set the tuning parameters as in \Cref{sec:implementation} with constants $c_\varepsilon=5$, $c_\eta=0.12$, and $c_\iota=0.3$.
All the simulation results are robust to a 20\% relative change in values of the tuning parameters.

Since the full bootstrap is computationally costly, we use the warp speed bootstrap instead \citep[e.g.,][]{giacomini2013warp}. The latter works by only drawing one bootstrap sample per Monte Carlo replication and pooling the resulting bootstrap statistics across replications to estimate the critical value. Finally, we compute the AME bounds by profiling over the resulting confidence regions.

\Cref{fig:mc} shows the average distance statistic and coverage probabilities across simulations at each grid point, along with the true identified set $\Theta_I$ circumscribed by the black line. On the left panel, the average distance is close to zero over the interior of the identified set and increases outside it, so the statistic tracks the geometry of the identified set reasonably well. This pattern also explains the robustness to deviations in the tuning parameters: moderate changes in the tuning parameters mainly affect the narrow region where $\hat D_\varepsilon(\theta^*)$ becomes substantially different from zero. On the right panel, coverage probabilities are close to one in the interior of the identified set and fall quickly outside it. \Cref{fig:mc-coverage-numbers} reports the values of the coverage probabilities and shows that the coverage on the boundary of the identified set is close to the nominal level of 90\%.

Finally, \Cref{tab:mc-ame} reports the median estimated identified intervals and median 90\% confidence regions for the AME of each covariate across simulations. We impose the logit derivative restriction in \Cref{app:logit-ame-attrition} before profiling. Although the true AMEs are positive, the intervals are wide and include zero for both covariates. This reflects the absence of restrictions on the attrition process, the additional conservativeness from profiling over the confidence region for $\theta$, and the fact that AME is partially identified in nonlinear models with fixed effects.

\begin{figure}[tp]
  \centering
  \includegraphics[width=\textwidth]{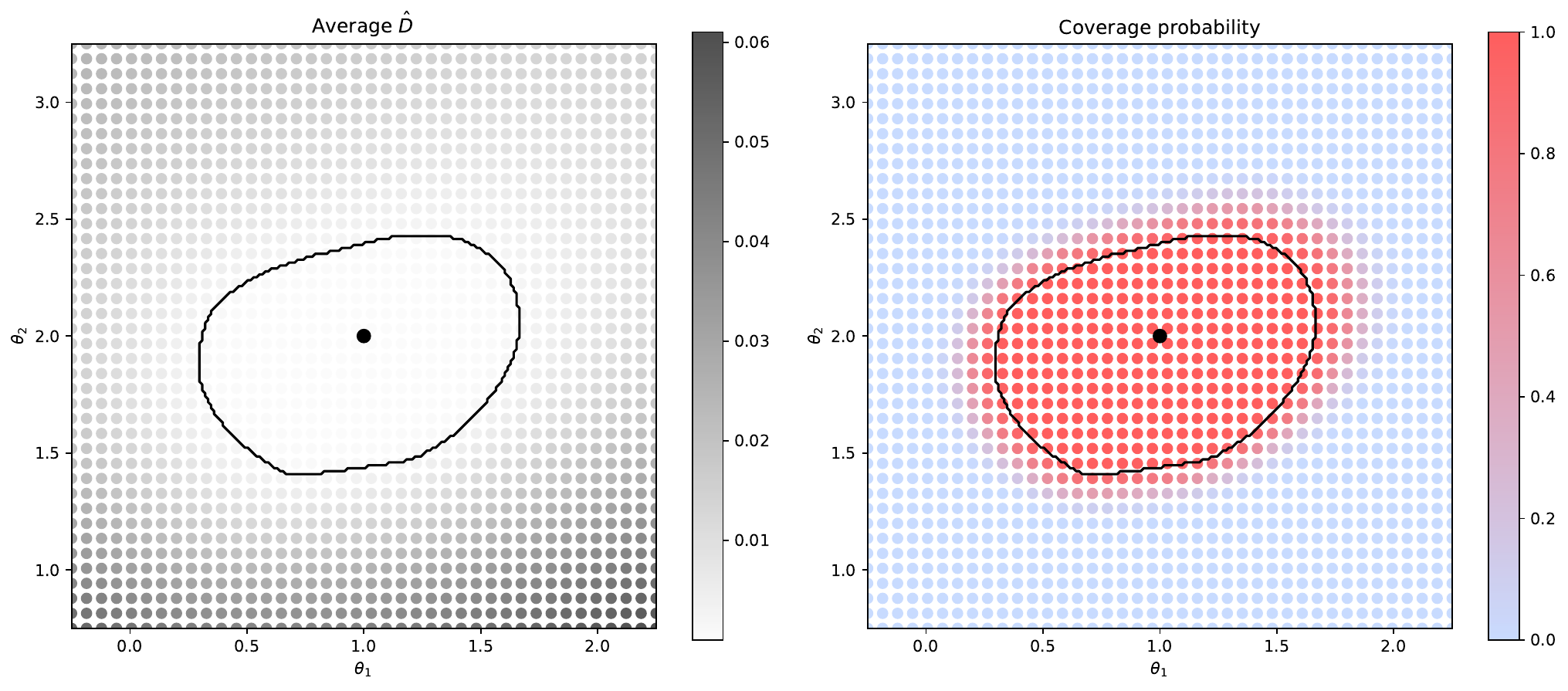}\vspace{-0.5cm}
  \caption{Left panel: average distance $\hat D_\varepsilon(\theta)$ across $500$ simulations. Right panel: coverage probability of $90\%$ confidence region. Black contour line: population identified set $\Theta_I$. Black dot: true value $\theta_0$.}
  \label{fig:mc}
\end{figure}

\begin{figure}[tp]
  \centering
  \includegraphics[width=0.6\textwidth]{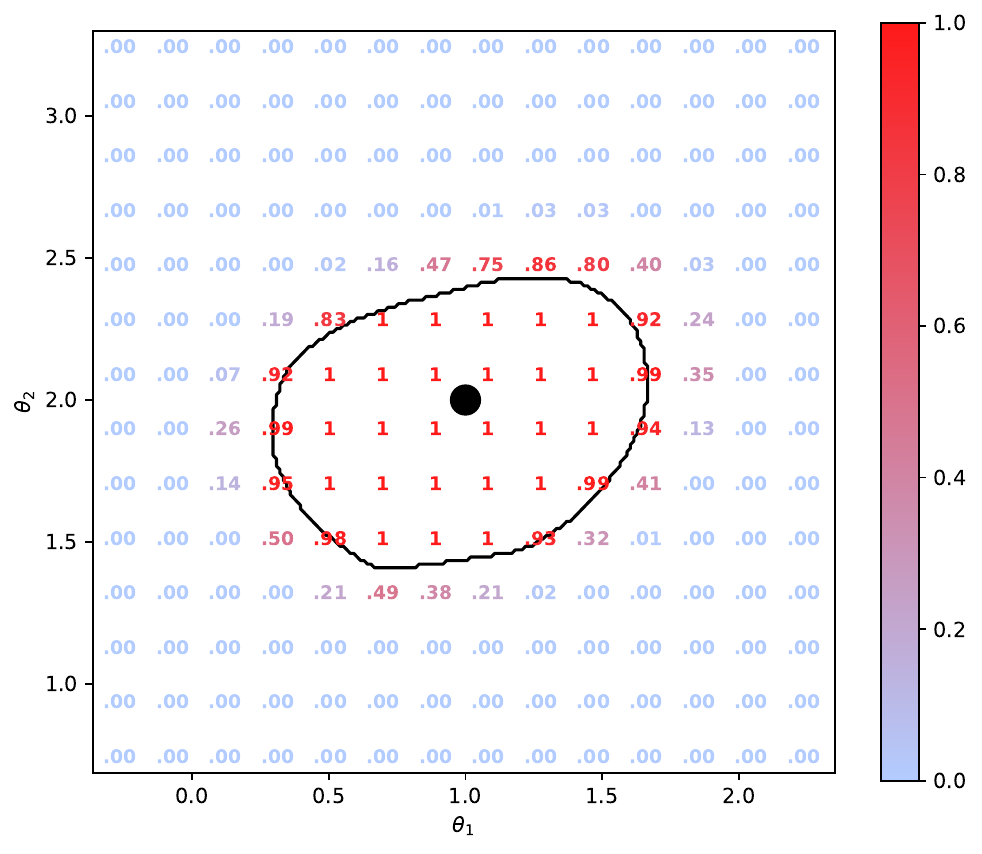}\vspace{-0.25cm}
  \caption{Coverage probability of $90\%$ confidence region for $\theta$ across $500$ simulations, with values displayed at grid points. Black contour line: population identified set $\Theta_I$. Black dot: true value $\theta_0$.}
  \label{fig:mc-coverage-numbers}
\end{figure}

\begin{table}[tp]
\centering
\bigskip
\renewcommand{\arraystretch}{1.2}
\begin{tabular}{lccc}
\hline\hline
Covariate & True & Est.\ id interval & 90\% CR \\
\hline
$X_1$ & $0.165$ & $[0.000,\ 0.434]$ & $[0.000,\ 0.450]$ \\
$X_2$ & $0.330$ & $[0.000,\ 0.620]$ & $[0.000,\ 0.636]$ \\
\hline
\end{tabular}
\caption{Median estimated identified intervals and 90\% confidence regions for the AME across $500$ simulations.}
\label{tab:mc-ame}
\end{table}

\section{Empirical illustration}\label{sec:empirical}

We study how self-reported happiness depends on agreeableness and extroversion when unobserved individual heterogeneity and panel attrition are both present. 
Specifically, we consider the panel logit regression
\begin{align*}
\text{happiness}_{it} = \mathbf{1} \left\{\alpha_i + \theta_1 \cdot\text{agreeableness}_{it} + \theta_2 \cdot \text{extroversion}_{it} - \varepsilon_{it} > 0 \right\},
\end{align*}
where $\alpha_i$ is an unobserved individual fixed effect. 
Panel data with fixed effects help control for differences in baseline happiness and reporting styles across individuals, but attrition may make retainer-only estimates sensitive to sample selection.
A refreshment sample provides additional information about the second-period distribution.
Our empirical example illustrates how findings based on the retainer-only (balanced) panel can change once attrition is taken into account while leaving its mechanism unrestricted.

We use survey data from the Understanding America Study (UAS), a large household panel collected and maintained by the USC Center for Economic and Social Research.
Specifically, we use the self-reported happiness, the agreeableness score, and the extroversion score for participants of waves 12 and 13 of the UAS.\footnote{See the \href{https://uasdata.usc.edu/page/Comprehensive+File+And+Panel+Dataset}{UAS Comprehensive Panel Dataset Description}, available on the UAS website, for information on the happiness score (sos002), the agreeableness score (pagreeableness), and the extroversion score (pextroversion).}
We set $\text{happiness}_{it}$ to $0$ for scores below the sample median and $1$ otherwise.
We group extroversion scores into 8 quantile bins labeled $0$ to $7$ and rescale the observed agreeableness scores to $[0,14]$ so that the two covariates have comparable variances.

The wave 12 sample contains $n_{\rm org}=6307$ respondents. Among them, $n_{\rm ret}=5131$ are retainers observed again in wave 13, so the attrition rate is $1-\hat p\approx 19\%$. We also observe a wave 13 refreshment sample of size $n_{\rm ref}=4449$. Among attritors, 53.4\% have $\text{happiness}_{i1}=1$. Among the refreshment sample, 53.5\% have $\text{happiness}_{i2}=1$. Among retainers, 43.2\% have $\text{happiness}_{i1}=\text{happiness}_{i2}=1$, 30.2\% have $\text{happiness}_{i1}=\text{happiness}_{i2}=0$, and 26.5\% switch happiness status across waves.
These patterns highlight both the importance of fixed effects, since most retainers have the same happiness status in both waves, and the relevance of attrition, since a nontrivial share of the original sample is not observed in wave 13.
The covariates also exhibit time variation: 89.0\% of retainers change agreeableness level and 64.3\% change extroversion level across waves.

We compare two approaches. The first procedure ignores the attrition problem and estimates a standard panel logit with fixed effects using the retainer-only panel. The second procedure acknowledges the possibility of nonrandom attrition and computes our OT-based partial identification bounds using the combined dataset.

In \Cref{tab:stayers}, we report the estimates and 90\% confidence intervals for $\theta$ under the first approach, together with the 90\% confidence regions under the second approach. In \Cref{fig:emp-app}, we plot the estimated identified set and the 90\% confidence region based on $B=500$ bootstrap samples, with the retainer-only confidence intervals overlaid for comparison.
In \Cref{tab:ame_bounds_ci}, we report the estimated identified sets and 90\% confidence regions for the AME of each covariate.
The retainer-only analysis gives positive estimates for both coefficients, with 90\% confidence intervals excluding zero. In contrast, the partial identification bounds are much wider and include zero. 
Therefore, the availability of the refreshment sample allows us to conclude that the positive relationship found in the retainer-only analysis breaks down when the strong assumption of random attrition is relaxed. 

\begin{table}[tp]
    \centering
    \bigskip
    \renewcommand{\arraystretch}{1.2}
    \begin{tabular}{lccc c}
    \hline\hline
     & \multicolumn{2}{c}{Attrition ignored} & & Attrition allowed\\
    \cline{2-3} \cline{5-5}
    Parameter & Estimate & $90\%$ CI & & $90\%$ CR \\
    \hline
    $\theta_1$ & $0.154$ & $[0.103,\ 0.205]$ & & $[-0.385,\ 0.769]$ \\
    $\theta_2$ & $0.067$ & $[0.012,\ 0.122]$ & & $[-0.538,\ 0.744]$ \\
    \hline
    \end{tabular}
    \caption{Parameter estimates and 90\% confidence intervals for the point-identified panel logit that ignores attrition and the 90\% confidence regions for the partially identified panel logit that allows for nonrandom attrition.}
    \label{tab:stayers}
\end{table}

\begin{figure}[tp]
    \centering
    \includegraphics[width=\textwidth]{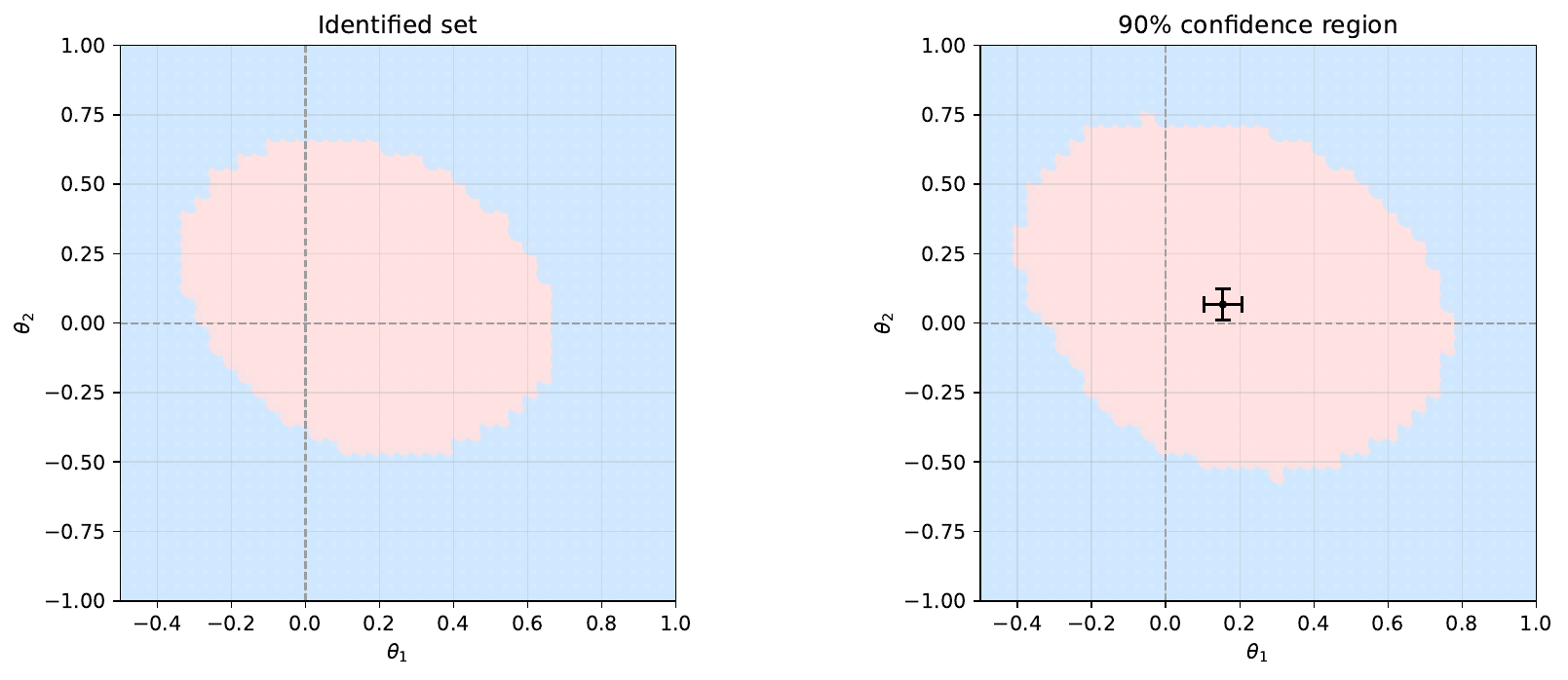}\vspace{-0.5cm}
    \caption{Left panel: estimated identified set for $\theta$ in the empirical illustration. Right panel: $90\%$ confidence region with $B=500$ bootstrap samples. Black intervals are $90\%$ confidence intervals that ignore attrition.}
    \label{fig:emp-app}
\end{figure}

\begin{table}[tp]
    \centering
    \bigskip
    \renewcommand{\arraystretch}{1.2}
    \begin{tabular}{lcc}
    \hline\hline
    Covariate & Identified set & 90\% CR \\
    \hline
    Agreeableness & $[-0.065,\ 0.141]$ & $[-0.081,\ 0.167]$ \\
    Extroversion  & $[-0.093,\ 0.138]$ & $[-0.115,\ 0.160]$ \\
    \hline
    \end{tabular}
    \caption{Estimated identified sets and 90\% confidence regions for the AME of each covariate.}
    \label{tab:ame_bounds_ci}
    \end{table}

\section{Conclusion}\label{sec:conclusion}

This paper develops a methodology for estimation and inference in a partially identified GMM model in which the joint distribution of the data is identified only up to its marginals.
We characterize the sharp identified set for the parameter of interest using optimal transport (OT) and use entropic regularization to construct a computationally tractable estimator.
We also propose a procedure for testing hypotheses and constructing confidence regions based on a test statistic that explicitly accounts for regularization.
We establish a novel uniform central limit theorem for the entropic OT value under smooth cost functions and use it to derive the asymptotic distribution of the test statistic.
We then obtain valid critical values via the bootstrap for directionally differentiable functionals of \citet{fang2019inference}.

Our methodology is broadly applicable, computationally efficient, and can be used to study economic questions ranging from the effects of policy interventions to the dynamics of household behavior.

Our framework admits several promising theoretical extensions.
First, it could be generalized to settings with more than two marginals, such as panel data with multiple waves or repeated cross-sections over several periods, using multi-marginal OT.
Second, it could be extended to models with conditional moment restrictions.
Finally, partial identification of a parameter of interest may arise both from incomplete knowledge of the joint distribution and from moment conditions that fail to identify the parameter uniquely even when that distribution is known.
The interaction between these two sources of underidentification is theoretically challenging and warrants further study.
We leave these extensions for future work.

\bibliographystyle{ecta}
\bibliography{references}

\newpage

\appendix
\renewcommand{\appendixpagename}{Online Appendix}
\appendixpage

\begin{subappendices}
\crefalias{section}{subappendix}
\crefalias{subsection}{subappendix}
\renewcommand{\thesection}{\Alph{section}}
\renewcommand{\thesubsection}{\thesection.\arabic{subsection}}

\section{Proofs of theoretical results}\label{app:proofs}

\subsection{Proof of \Cref{thm:id-set-characterization}}\label{app:id-set-characterization}

    We follow the proof of Proposition 1 in \citet{franguridi2025set} closely.
    By \Cref{a:correct-specification}, $\Theta_{I}$ is nonempty.
    
  \textbf{Step 1: $\mathcal{M}(\theta)$ is compact and convex.}
  
  \noindent
  Convexity is trivial.
  In view of \Cref{a:supports-compact,a:phi-continuous}, Theorem 15.11 and Corollary 15.7 in \citet{aliprantis2006infinite} imply that $\Pi$ is compact in the weak topology and the map $\pi \mapsto \E_{\pi}\phi(X,Y,\theta)$ is continuous.
  Hence, $\mathcal{M}(\theta)$ is compact as the image of a compact set under a continuous map.
  
  \noindent
  \textbf{Step 2: $\Theta_{I} = D^{-1}(\{0\})$.}
  
  \noindent
  Since $\mathcal{M}(\theta)$ is closed and convex, $0\in \mathcal{M}(\theta)$ if and only if its support function is everywhere nonnegative
  \begin{align*}
    \psi_{\mathcal{M}(\theta)}(u) \bydef \max_{\pi\in\Pi} u'\E_\pi \phi(X,Y,\theta) \ge 0  \text{ for all } u\in \R^{\dim(\phi)}.
  \end{align*}
  Indeed, if $0\in \mathcal{M}(\theta)$, then $\psi_{\mathcal{M}(\theta)}(u) = \max_{\nu\in \mathcal{M}(\theta)} u'\nu \ge 0$.
  Conversely, if $0\notin \mathcal{M}(\theta)$, then by strong separation of a point from the closed convex set $\mathcal{M}(\theta)$, there exists $u\neq 0$ and $\alpha>0$ such that $u'\nu\le -\alpha$ for all $\nu\in \mathcal{M}(\theta)$.
  This implies $\psi_{\mathcal{M}(\theta)}(u) \le -\alpha < 0$.
  
  Since $\psi_{\mathcal{M}(\theta)}(0)=0$ for all $\theta$, nonnegativity of $\psi_{\mathcal{M}(\theta)}$ is equivalent to the equality of
  \begin{align}
      \min_{u\in \R^{\dim(\phi)}} \psi_{\mathcal{M}(\theta)}(u) \label{eq:min-psi-Rd}
  \end{align}
  to zero.
  Let us show that the latter is equivalent to the equality of
  \begin{align}
      \min_{u\in \B} \psi_{\mathcal{M}(\theta)}(u) \label{eq:min-psi-B}   
  \end{align}
  to zero.
  Indeed, if \eqref{eq:min-psi-Rd} is zero, then the minimum is achieved at $u=0$, and hence \eqref{eq:min-psi-B} is zero.
  Conversely, if \eqref{eq:min-psi-Rd} is nonzero, then there exists $u$ such that $\psi_{\mathcal{M}(\theta)}(u) < 0$. By the positive homogeneity of the support function, $\psi_{\mathcal{M}(\theta)}(u/\|u\|) = \psi_{\mathcal{M}(\theta)}(u)/\|u\|<0$, and hence the expression \eqref{eq:min-psi-B} is negative.
  
  \noindent
  \textbf{Step 3: $\Theta_{I}$ is compact.}
  
  \noindent
  Since $\Theta_{I} = D^{-1}(\{0\})$, it suffices to establish the continuity of $D(\theta)$.
  For this, notice that
  \begin{align*}
      (u,\theta,\pi) \mapsto u' \E_{\pi}\phi(X,Y,\theta)
  \end{align*}
  is a continuous function on the compact set $\B \times \Theta\times \Pi$, where $\Pi$ is equipped with the weak topology.
  Berge's maximum theorem (see, e.g., Theorem 17.31 in \citet{aliprantis2006infinite}) implies that the function
  \begin{align*}
      (u,\theta) \mapsto \max_{\pi\in\Pi}  u'\E_{\pi}\phi(X,Y,\theta) = \psi_{\mathcal{M}(\theta)}(u)
  \end{align*}
  is continuous on the compact set $\B \times \Theta$.
  Applying Berge's theorem again implies that $D(\theta)$ is continuous, completing the proof.

\subsection{Proof of \Cref{thm:eot-stability}}\label{app:eot-stability}

We rely on an idea of shadow coupling by \citet[Lemma 3.2]{eckstein2022quantitative} and follow the arguments in the proof of their Theorem 3.7.

Let $\tilde\mu$ and $\tilde\nu$ be probability measures on $\mathcal X$ and $\mathcal Y$, respectively.
Fix $\varepsilon\geq0$, $\theta\in\Theta$,
and $u$ with $\|u\|\leq1$.
Define
\[
f_{u,\theta}(x,y) = u'\phi(x,y,\theta).
\]
By \Cref{a:lipschitz-moments},
\begin{equation*}
|f_{u,\theta}(x,y)-f_{u,\theta}(x',y')|
\leq
L_1\|x-x'\|+L_2\|y-y'\|.
\end{equation*}
Let $\gamma_X\in\Pi(\mu,\widetilde\mu)$
and $\gamma_Y\in\Pi(\nu,\widetilde\nu)$
be optimal couplings for $W_1$, which exist because the sets $\mathcal{X}$ and $\mathcal{Y}$ are compact. By disintegration,
\[
\gamma_X(dx,dx')=\mu(dx)K_X(x,dx'),
\qquad
\gamma_Y(dy,dy')=\nu(dy)K_Y(y,dy'),
\]
where $K_X$ and $K_Y$ are Markov kernels.
Define a new probability measure $\widetilde{\pi}$ by
\[
\widetilde\pi(A\times B)
=
\int_{\mathcal X\times\mathcal Y}
K_X(x,A)K_Y(y,B)\,\pi(dx,dy).
\]
Since $K_X(x,\mathcal{X})=K_Y(\mathcal{Y},y)=1$, we have $\widetilde\pi\in\Pi(\widetilde\mu,\widetilde\nu)$.
Moreover,
\begin{align}
\left|
\int f_{u,\theta}\,d\widetilde\pi
-
\int f_{u,\theta}\,d\pi
\right|
&\leq
L_1W_1(\mu,\widetilde\mu)
+
L_2W_1(\nu,\widetilde\nu)
=:\delta.
\label{eq:cost_difference}
\end{align}
By the data-processing inequality \citep[][Lemma 4.1]{eckstein2022quantitative},
\begin{equation}
\operatorname{KL}
(\widetilde\pi\|
\widetilde\mu\otimes\widetilde\nu)
\leq
\operatorname{KL}
(\pi\|\mu\otimes\nu).
\label{eq:kl_contraction}
\end{equation}
Combining \eqref{eq:cost_difference} and \eqref{eq:kl_contraction} yields
\begin{align*}
\int f_{u,\theta}\,d\widetilde\pi
+ \varepsilon
\operatorname{KL}
(\widetilde\pi\|
\widetilde\mu\otimes\widetilde\nu)
\leq
\int f_{u,\theta}\,d\pi
+
\varepsilon
\operatorname{KL}(\pi\|\mu\otimes\nu)
+\delta.
\end{align*}

Taking the infimum over
$\pi\in\Pi(\mu,\nu)$ gives
\[
c_{\theta,\varepsilon}
(u;\widetilde\mu,\widetilde\nu)
\leq
c_{\theta,\varepsilon}(u;\mu,\nu)+\delta,
\]
Repeating the argument with the two pairs
of marginals interchanged yields
\[
c_{\theta,\varepsilon}(u;\mu,\nu)
\leq
c_{\theta,\varepsilon}
(u;\widetilde\mu,\widetilde\nu)
+\delta.
\]

Therefore,
\begin{equation*}
\left|
c_{\theta,\varepsilon}
(u;\widetilde\mu,\widetilde\nu)
-
c_{\theta,\varepsilon}(u;\mu,\nu)
\right|
\leq\delta.
\end{equation*}
Using this bound, we obtain
\begin{align*}
\left|
D_\varepsilon
(\theta;\widetilde\mu,\widetilde\nu) - D_\varepsilon(\theta;\mu,\nu)
\right|
&= \left|
\max_{\|u\|\le 1} c_{\theta,\varepsilon}
(u;\widetilde\mu,\widetilde\nu) - \max_{\|u\|\le 1} c_{\theta,\varepsilon}(u;\mu,\nu) \right| \\
&\le \max_{\|u\|\le 1} \left| c_{\theta,\varepsilon}(u;\widetilde\mu,\widetilde\nu) - c_{\theta,\varepsilon}(u;\mu,\nu) \right| \\
&\leq\delta.
\end{align*}
Plugging in the empirical measures $\tilde\mu=\hat\mu$ and $\tilde\nu=\hat\nu$ and using the fact that they converge in the Wasserstein distance at a known rate \citep[e.g.,][]{fournier2015rate} completes the proof.

\subsection{Proof of \Cref{thm:eot-bias}}\label{app:eot-bias}

We use the block-approximation argument of \citet[Proposition 3.1]{carlier2023convergence} with an explicit bound on the number of partition cells.

Fix $0<\varepsilon\leq1$, $\theta\in\Theta$, and $\|u\|\leq1$, and write
$f_{u,\theta}(x,y)=u'\phi(x,y,\theta)$.
The nonnegativity of the Kullback--Leibler divergence implies $c_{\theta,\varepsilon}(u)\geq c_{\theta}(u)$.

We first derive the upper bound by partitioning the
support of $\mu$.

\medskip
\noindent
\textit{Step 1: Construction of a finite partition.}

Fix $\delta>0$. Partition an axis-aligned cube containing $\mathcal{X}$ into smaller cubes of Euclidean diameter at most $\delta$.
Such a partition can be constructed by dividing each coordinate interval into
\[
m_1=
\max\left\{
1,\left\lceil\frac{B_1\sqrt d}{\delta}
\right\rceil
\right\}
\]
equal subintervals. 
Let $\{A_i\}_{i=1}^N$ denote the resulting partition
of $\mathcal{X}$ (omitting empty cells). Then $\operatorname{diam}(A_i)\leq\delta$ for every $i$, and the number of cells satisfies
\begin{equation}
N\leq m_1^d
\leq
\left(1+\frac{B_1\sqrt d}{\delta}\right)^d.
\label{eq:covering}
\end{equation}

\medskip
\noindent
\textit{Step 2: Construction of an approximate coupling.}

Let $\pi_0\in\Pi(\mu,\nu)$ be an optimal coupling for the unregularized problem, which exists by the compactness of the supports and continuity of the cost.
For each cell, let $p_i=\mu(A_i)$ and define
\[
\pi^\delta(dz_1,dz_2)
=
\sum_{i:\, p_i>0}
\frac{\mu(dz_1\cap A_i)}{p_i}
\,\pi_0(A_i,dz_2).
\]
Notice that $\pi^\delta\in\Pi(\mu,\nu)$.

Equivalently, draw $(X,Y)\sim\pi_0$ and,
conditional on the cell $A_i$ containing $X$,
draw $X'$ independently from $\mu(\cdot\mid A_i)$
while keeping $Y$ unchanged. Then $(X',Y)$ has distribution $\pi^\delta$.

Since $\|X'-X\|\leq\delta$ by construction, \Cref{a:lipschitz-moments} and $\|u\|\leq1$ imply
\begin{align*}
&\int f_{u,\theta}(z_1,z_2)\,d\pi^\delta
-\int f_{u,\theta}(z_1,z_2)\,d\pi_0 =
\mathbb E[f_{u,\theta}(X',Y)-f_{u,\theta}(X,Y)]
\leq
L_1\mathbb E\|X'-X\|
\leq L_1\delta.
\end{align*}
Thus,
\begin{equation}
\int f_{u,\theta}(x,y)\,d\pi^\delta(x,y)-c_{\theta}(u)
\leq L_1\delta.
\label{eq:cost-error}
\end{equation}

\medskip
\noindent
\textit{Step 3: Bounding the KL divergence.}

For every $i$ with $p_i>0$, define
\[
\nu_i(dz_2)=\frac{\pi_0(A_i,dz_2)}{p_i}.
\]
Then
\[
\pi^\delta
=
\sum_{i:\, p_i>0}
p_i\,\mu(\cdot\mid A_i)\otimes\nu_i.
\]

Since
\[
\nu=\sum_{i:\, p_i>0} p_i\nu_i,
\]
we have $p_i\nu_i\leq\nu$ as measures.
Therefore, $\nu_i\ll\nu$ and
\[
\frac{d\nu_i}{d\nu}\leq\frac1{p_i}
\qquad \nu\text{-a.e.}
\]

Writing $f_i=d\nu_i/d\nu$ for $p_i>0$, we have
$\pi^\delta\ll\mu\otimes\nu$ with density
\[
\frac{d\pi^\delta}{d(\mu\otimes\nu)}(z_1,z_2)
=\sum_{i:\, p_i>0}\mathbf{1}_{A_i}(z_1)f_i(z_2)
\qquad (\mu\otimes\nu)\text{-a.e.}
\]
Since the sets $A_i$ are disjoint,
\begin{align*}
\mathrm{KL}(\pi^\delta\|\mu\otimes\nu)
&=
\sum_{i:\, p_i>0} p_i\,\mathrm{KL}(\nu_i\|\nu)
\leq
\sum_{i:\, p_i>0}p_i\log(1/p_i)
\leq\log N.
\end{align*}

Using \eqref{eq:covering}, we obtain
\begin{equation}
\mathrm{KL}(\pi^\delta\|\mu\otimes\nu)
\leq
d\log\left(
1+\frac{B_1\sqrt d}{\delta}
\right).
\label{eq:entropy-error}
\end{equation}

\medskip
\noindent
\textit{Step 4: Combining the bounds.}

Evaluating the EOT objective
at $\pi^\delta$ and using
\eqref{eq:cost-error}--\eqref{eq:entropy-error},
we obtain
\begin{align*}
c_{\theta,\varepsilon}(u)-c_{\theta}(u)
&\leq
\int f_{u,\theta}(x,y)\,d\pi^\delta(x,y)-c_{\theta}(u)
+\varepsilon
\mathrm{KL}(\pi^\delta\|\mu\otimes\nu)\\
&\leq
L_1\delta+
d\varepsilon\log\left(
1+\frac{B_1\sqrt d}{\delta}
\right).
\end{align*}

Repeating the same construction with a partition of
$\mathcal{Y}$ gives
\[
c_{\theta,\varepsilon}(u)-c_{\theta}(u)
\leq
L_2\delta+
d\varepsilon\log\left(
1+\frac{B_2\sqrt d}{\delta}
\right).
\]
Taking the minimum of these two bounds and choosing $\delta=\varepsilon$ yields
\[
c_{\theta,\varepsilon}(u)-c_{\theta}(u)
\leq
\varepsilon\min_{j=1,2}
\left\{
L_j+d\log\left(
1+\frac{B_j\sqrt d}{\varepsilon}
\right)
\right\}.
\]
Since $0<\varepsilon\leq1$, we have
\[
1+\frac{B_j\sqrt d}{\varepsilon}
=
\frac{\varepsilon+B_j\sqrt d}{\varepsilon}
\leq
\frac{1+B_j\sqrt d}{\varepsilon}.
\]
Taking logarithms, substituting the resulting inequality into the preceding bound, and taking maxima over $\|u\|\leq1$ completes the proof.

\subsection{Proof of \Cref{thm:consistency}}\label{app:consistency}

Define
\[
a_n=
\sup_{\theta\in\Theta}
|\widehat D_{\varepsilon_n}(\theta)-D_{\varepsilon_n}(\theta)|,
\qquad
b_n=\sup_{\theta\in\Theta_I}D_{\varepsilon_n}(\theta).
\]
By \Cref{a:consistency-tuning-parameters}, we have $(a_n+b_n)/\eta_n=o_p(1).$
Therefore, the event
\[
E_n=\{a_n+b_n\leq\eta_n\}
\]
satisfies \(\Pr(E_n)\to1\).
For any \(\theta\in\Theta_I\), on \(E_n\),
\[
\widehat D_{\varepsilon_n}(\theta)
\leq D_{\varepsilon_n}(\theta)+a_n
\leq b_n+a_n
\leq\eta_n.
\]
Hence $\Theta_I\subseteq\widehat\Theta_I$ on $E_n$.

Next, for any \(\theta\in\widehat\Theta_I\),
\[
D(\theta)
\leq D_{\varepsilon_n}(\theta)
\leq\widehat D_{\varepsilon_n}(\theta)+a_n
\leq\eta_n+a_n.
\]
On \(E_n\), \(a_n\leq\eta_n\), so
\begin{align}
    \Theta_I
    \subseteq\widehat\Theta_I
    \subseteq\{\theta\in\Theta:\,D(\theta)\leq2\eta_n\}. \label{eq:policeman-bounds}
\end{align}
Define
\[
\omega_0(s)
=
\sup_{\theta\in\Theta:\,D(\theta)\leq s}
d(\theta,\Theta_I).
\]
Since \(\Theta\) is compact, $\Theta_I$ is nonempty, and \(D\) is continuous (see the proof of \Cref{thm:id-set-characterization}), we have $\omega_0(s)\to0$ as $s\downarrow0$.
Indeed, otherwise there would exist \(s_k\downarrow0\) and \(\theta_k\in\Theta\) with \(D(\theta_k)\leq s_k\) but \(d(\theta_k,\Theta_I)\geq\delta>0\). A convergent subsequence would have a limit \(\theta^*\) satisfying \(D(\theta^*)=0\), contradicting the distance bound.

By \eqref{eq:policeman-bounds}, on \(E_n\),
\[
d_H(\widehat\Theta_I,\Theta_I)
\leq\omega_0(2\eta_n).
\]
Taking into account \(\omega_0(2\eta_n)\to0\) and \(\Pr(E_n)\to1\) completes the proof.

\subsection{Proof of \Cref{thm:regularized-identified-set}}\label{app:regularized-identified-set}

To prove the set inclusion, let $\theta\in\Theta_{I}$ and notice that
$D(\theta)=0$ and $
D_\varepsilon(\theta)
\leq D(\theta)+r(\varepsilon)
=r(\varepsilon)$.
Hence, $\theta\in\Theta_{I,\varepsilon}$.

To prove the Hausdorff convergence, let $\theta\in\Theta_{I,\varepsilon}$.
We have
\[
D(\theta) \le D_{\varepsilon}(\theta) \leq r(\varepsilon).
\]
Combining this inequality with
\eqref{eq:identification} yields
\[
\kappa\,
d(\theta,\Theta_{I})^\alpha
\leq r(\varepsilon).
\]
Therefore,
\[
d(\theta,\Theta_{I})
\leq
\left(\frac{r(\varepsilon)}{\kappa}
\right)^{1/\alpha}.
\]
Taking the supremum over
$\theta\in\Theta_{I,\varepsilon}$ gives
\begin{equation}
\sup_{\theta\in\Theta_{I,\varepsilon}}
d(\theta,\Theta_{I})
\leq
\left(\frac{r(\varepsilon)}{\kappa}
\right)^{1/\alpha}.
\label{eq:one_sided_bound}
\end{equation}

Recall that the Hausdorff distance is
\[
d_H(A,B)
=
\max\left\{
\sup_{a\in A}d(a,B),
\sup_{b\in B}d(b,A)
\right\}.
\]
As shown above, $\Theta_{I}\subseteq
\Theta_{I,\varepsilon}$, and so
\[
\sup_{\theta_0\in\Theta_{I}}
d(\theta_0,\Theta_{I,\varepsilon})
=0.
\]
Therefore,
\[
d_H(\Theta_{I,\varepsilon},\Theta_{I})
=
\sup_{\theta\in\Theta_{I,\varepsilon}}
d(\theta,\Theta_{I}).
\]
Combining this equality with \eqref{eq:one_sided_bound} completes the proof.

\subsection{Proof of \Cref{thm:UCLT}}\label{app:UCLT}

We set $\varepsilon=1$ for simplicity.
For a function $f \in C^s(\X\times\Y)$, denote its H\"older norm by
\begin{align*}
    \|f\|_{C^s(\X\times\Y)} = \max_{0 \le |\alpha| \le s} \sup_{(x,y) \in \X \times \Y} \left| \nabla^\alpha f(x,y) \right|,
\end{align*}
where the maximum is taken over all multi-indices $\alpha=(\alpha_1,\dots,\alpha_{2d}) \in \N_0^{2d}$ of order $|\alpha|\bydef \alpha_1+\cdots+\alpha_{2d} \le s$, and the partial derivative operator
\begin{align*}
    \nabla^\alpha f(x,y) \bydef \frac{\partial^{|\alpha|}}{\partial x_1^{\alpha_1} \cdots \partial x_d^{\alpha_d} \partial y_1^{\alpha_{d+1}} \cdots \partial y_d^{\alpha_{2d}}} f(x,y).
\end{align*}
Assumptions \ref{a:param-compact}, \ref{a:supports-compact}, and \ref{a:UCLT} imply that there exist finite constants $\bar\phi$ and $\bar \phi_s$ such that 
\begin{align*}
    \sup_{u\in \B} \sup_{\theta\in\Theta} \sup_{(x,y)\in\X\times\Y} |u'\phi(x,y,\theta)| \le \bar \phi
    \end{align*}
 and
\begin{align*}
\sup_{u\in \B} \sup_{\theta\in\Theta} \|u'\phi(\cdot,\cdot,\theta)\|_{C^s (\X \times \Y)} \le \bar \phi_s.
\end{align*}
Define the mapping $c: \ell^\infty(\F) \to \ell^\infty(\B\times\Theta)$ by
\begin{align*}
    c(\mu)(u,\theta) = \sup_{\varphi,\psi} \int \varphi \, d\mu + \int \psi \,d \nu - \int e^{\varphi \oplus \psi - u'\phi(\cdot,\cdot,\theta)}\,d \mu\otimes\nu + 1
\end{align*}

\begin{proof}

    \textbf{Part (i)}. Since $s > d/2$, the class
        \begin{align*}
            \mathscr{F} = \left\{ f \in C^s(\X) \text{ such that }\|f\|_{C^s} \le C_{s,d} \right\},
        \end{align*}
        where $C_{s,d}$ is the constant in \Cref{prop:potentials-deriv-bounded}, is Donsker, see, e.g., Corollary 2.7.2 in \citet{vaart2023empirical}.
        Therefore,
        \[
        \sqrt{n}(\hat\mu-\mu) \weakto \G \text{ in } \ell^\infty(\mathscr{F}),
        \]
        where $\G$ is the generalized $\mu$-Brownian bridge.
        In view of \Cref{prop:hadamard}, we can apply the delta method for random measures in Proposition 1 of \citet{goldfeld2024statistical} to obtain
        \begin{align*}
            \sqrt n(\theta(\hat\mu)-\theta(\mu)) \weakto \theta_{\mu}'(\G) = \G(\varphi)\sim N(0,\V_{\mu}(\varphi)).
        \end{align*}
        Since $\theta_{\mu}'$ is the point evaluation at $\varphi$, Proposition 2 in \citet{goldfeld2024statistical} implies that $\theta(\hat\mu)$ is semiparametrically efficient.
    
        Finally, since $\theta_{\mu}'$ is linear, $\rho \mapsto \theta(\rho)$ is (fully) Hadamard differentiable at $\rho=\mu$ tangentially to $\supp \G$ by Corollary 1 in \citet{goldfeld2024statistical}. Hence, nonparametric bootstrap is consistent by Theorem 23.9 in \citet{van2000asymptotic}.
        
    \textbf{Part (ii)}. Define the function class
    \begin{align*}
        \mathscr{F}^{\oplus} = \{\varphi\oplus \psi: (\varphi,\psi) \text{ satisfies } \eqref{eq:phi-deriv-bound}, \eqref{eq:psi-deriv-bound}  \}.
    \end{align*}
    Let us show that
    \begin{align}
        |\theta(\mu_1,\nu_1)-\theta(\mu_0,\nu_0)| \le \| \mu_1 \otimes \nu_1 - \mu_0\otimes \nu_0\|_{\mathscr{F}^\oplus} \label{eq:theta-Lipschitz}
    \end{align}
    for any probability measures $\mu_0,\mu_1$ on $\X$ and probability measures $\nu_0,\nu_1$ on $\Y$.
    Let $\varphi_{ij}, \psi_{ij}$ be optimal potentials for $(\mu_i,\nu_j)$ satisfying \eqref{eq:phi-deriv-bound}, \eqref{eq:psi-deriv-bound}.
    Then \eqref{eq:theta-lower} implies
    \begin{align*}
        \theta(\mu_1,\nu_1)-\theta(\mu_0,\nu_0) &= \theta(\mu_1,\nu_1)-\theta(\mu_0,\nu_1)+\theta(\mu_0,\nu_1)-\theta(\mu_0,\nu_0) \\
        &\ge \int \varphi_{01} \, d(\mu_1-\mu_0)+\int \psi_{00}\, d(\nu_1-\nu_0) \\
        &= \int (\varphi_{01}\oplus \psi_{00}) \, d( \mu_1\otimes \nu_1-\mu_0\otimes\nu_0).
    \end{align*}
    Similarly, \eqref{eq:theta-upper} implies
    \begin{align*}
        \theta(\mu_1,\nu_1)-\theta(\mu_0,\nu_0) \le \int (\varphi_{11}\oplus \psi_{01}) \, d( \mu_1\otimes \nu_1-\mu_0\otimes\nu_0).
    \end{align*}
    Since $\varphi_{ij}\oplus \psi_{k\ell} \in \mathscr{F}^\oplus$, we obtain \eqref{eq:theta-Lipschitz}.
    
    Moreover, arguing as in the proof of \Cref{prop:hadamard}, we obtain
    \begin{align*}
        \lim_{t\downarrow 0} \frac{\theta(\mu_0+t(\mu_1-\mu_0), \nu_0+t(\nu_1-\nu_0)) - \theta(\mu_0,\nu_0)}{t} = \int (\varphi_{00}\oplus \psi_{00}) \, d( \mu_1\otimes \nu_1-\mu_0\otimes\nu_0).
    \end{align*}
    
    Define $\mathscr{P}_0$ as the set of probability measures of the form $\rho_1\otimes\rho_2$, where $\rho_1$, $\rho_2$ concentrate in $\X$, $\Y$, respectively.
    Applying Proposition 1 in \citet{goldfeld2024statistical} to the statement in \Cref{prop:Donsker-joint} for $\delta(\rho_1\otimes \rho_2)=\theta(\rho_1,\rho_2)$ and $\mathscr{F}=\mathscr{F}^\oplus$ yields
    \begin{align*}
        \sqrt n(\theta(\hat\mu_n,\hat\nu_n)-\theta(\mu,\nu))&=\sqrt{n}(\delta(\hat\mu_n\otimes \hat\nu_n)-\delta(\mu\otimes \nu)) \\
        &\weakto \delta'_{\mu\otimes \nu}(\G_{\mu\otimes\nu})=\G_{\mu\otimes\nu}(\varphi\oplus\psi) \sim N(0, \V_\mu(\varphi)+\V_\nu(\psi)).
    \end{align*}
    Finally, $\V_\mu(\varphi)+\V_\nu(\psi)$ is the semiparametric variance bound due to Corollary 2 of \citet{goldfeld2024statistical}.
\end{proof}


The following result and its proof follow Proposition A.1 in \citet{mena2019statistical}.

\begin{proposition}\label{prop:potentials-bounded}
    For any $u\in \B$ and $\theta\in\Theta$, there exist optimal potentials $\varphi_{u,\theta}, \psi_{u,\theta}$ such that $\varphi_{u,\theta} \in C^s(\X)$, $\psi_{u,\theta} \in C^s(\Y)$, and 
    \begin{align*}
        |\varphi_{u,\theta}(x)| \le \bar\phi, \\
        |\psi_{u,\theta}(y)| \le \bar\phi,
    \end{align*}
    for all $x\in\X$ and $y\in\Y$.
\end{proposition}

\begin{proof}
    Fix $u\in\B$ and $\theta\in\Theta$ and let $(\varphi_{u,\theta}^0,\psi_{u,\theta}^0)$ be any pair of optimal potentials. Assume without loss of generality that $u'\phi(\cdot,\cdot,\theta)\ge 0$. 
    Since $(\varphi_{u,\theta}^0-a,\psi_{u,\theta}^0+a)$ is also a pair of optimal potentials for any $a\in\R$, we can assume
    \[
    \int \varphi_{u,\theta}^0(x)\, d\mu(x) = \int \psi_{u,\theta}^0(y)\, d\nu(y) = \frac{1}{2} c(\mu,\nu)(u,\theta) \ge 0.
    \]
    Define
    \begin{align*}
        \varphi_{u,\theta}(x)= - \log \int e^{\psi_{u,\theta}^0(y)-u'\phi(x,y,\theta)} \, d\nu(y), \\
        \psi_{u,\theta}(y) = - \log \int e^{\varphi_{u,\theta}^0(x)-u'\phi(x,y,\theta)} \, d\mu(x).
    \end{align*}
    Jensen's inequality combined with $\int \psi_{u,\theta}^0(y)\, d\nu(y) \ge 0$ yields
    \begin{align*}
        \varphi_{u,\theta}(x) \le - \int (\psi_{u,\theta}^0(y)-u'\phi(x,y,\theta))\, d\nu(y) \le \int u'\phi(x,y,\theta)\, d\nu(y) \le \bar \phi.
    \end{align*}
    To establish the lower bound on $\varphi_{u,\theta}(x)$, notice that, by Jensen's inequality,
    \begin{align*}
        \psi_{u,\theta}^0(y) &= - \log \int e^{\varphi_{u,\theta}^0(x)-u'\phi(x,y,\theta)} \, d\mu(x) \\
        &\le - \int \varphi_{u,\theta}^0(x)\, d\mu(x) + \int u'\phi(x,y,\theta) \, d\mu(x) \\
        &\le \int u'\phi(x,y,\theta) \, d\mu(x).
    \end{align*}
    Therefore,
    \begin{align*}
        \exp\{\psi_{u,\theta}^0(y) - u'\phi(x,y,\theta)\} \le \exp\left\{\int u'\phi(x,y,\theta) \, d\mu(x) - u'\phi(x,y,\theta)\right\} \le \exp\{\bar\phi\}.
    \end{align*}
    Integrating this inequality w.r.t.\ $\nu$ and taking $-\log$ yields
    \begin{align*}
        \varphi_{u,\theta}(x) \ge - \bar \phi.
    \end{align*}
    
    Notice that $\varphi_{u,\theta} \in C^s(\X)$ and $\psi_{u,\theta} \in C^s(\Y)$ by the dominated convergence theorem.
    
    It remains to show that $\varphi_{u,\theta},\psi_{u,\theta}$ are optimal potentials. By construction,
    \begin{align*}
        \int e^{\varphi_{u,\theta}(x)+\psi_{u,\theta}(y)-u'\phi(x,y,\theta)} \, d\mu(x) = 1 \text{ for all } y \in \Y.
    \end{align*}
    Moreover,
    \begin{align*}
        \int e^{\varphi_{u,\theta}(x)+\psi_{u,\theta}(y)-u'\phi(x,y,\theta)} \, d\mu(x) d\nu(y) &= \int e^{\varphi_{u,\theta}(x)+\psi_{u,\theta}^0(y)-u'\phi(x,y,\theta)} \, d\mu(x) d\nu(y) \\
        &= \int e^{\varphi_{u,\theta}^0(x)+\psi_{u,\theta}^0(y)-u'\phi(x,y,\theta)} \, d\mu(x) d\nu(y).
    \end{align*}
    By Jensen's inequality,
    \begin{align*}
        &\int (\varphi_{u,\theta}(x)-\varphi_{u,\theta}^0(x))\, d\mu(x) + \int (\psi_{u,\theta}(y)-\psi_{u,\theta}^0(y))\, d\nu(y) \\
        &\ge -\log \int e^{\varphi_{u,\theta}(x)-\varphi_{u,\theta}^0(x)} \, d\mu(x) -\log \int e^{\psi_{u,\theta}(y)-\psi_{u,\theta}^0(y)} \, d\nu(y) \\
        &= - \log \int e^{\varphi_{u,\theta}^0(x)+\psi_{u,\theta}^0(y)-u'\phi(x,y,\theta)} \, d\mu(x) d\nu(y) -
        \log \int e^{\varphi_{u,\theta}(x)+\psi_{u,\theta}^0(y)-u'\phi(x,y,\theta)} \, d\mu(x) d\nu(y) \\
        &=0.
    \end{align*}
    Since $(\varphi_{u,\theta}^0,\psi_{u,\theta}^0)$ maximizes the dual objective, so does $(\varphi_{u,\theta},\psi_{u,\theta})$. Therefore, $\varphi_{u,\theta},\psi_{u,\theta}$ are optimal potentials.
\end{proof}


The following result and its proof follow Proposition 1 in \citet{mena2019statistical}.

\begin{proposition}\label{prop:potentials-deriv-bounded}
    For any $u\in\B$ and $\theta\in\Theta$, there exist optimal dual potentials $\varphi_{u,\theta},\psi_{u,\theta}$ such that for any multi-index $\alpha \in \N_0^d$ of order $|\alpha|\le s$, 
    \begin{align}
        \left| \nabla^\alpha \varphi_{u,\theta}(x) \right| \le C_{s,d} \text{ for all } x\in\X, \label{eq:phi-deriv-bound} \\
        \left| \nabla^\alpha \psi_{u,\theta}(y) \right| \le C_{s,d} \text{ for all } y\in\Y, \label{eq:psi-deriv-bound}
    \end{align}
    where $C_{s,d}<\infty$ is a quantity that does not depend on $x,y,u,\theta$, but may depend on $s,d,\bar \phi$, and $\bar \phi_s$.
\end{proposition}

\begin{proof}
    Let $\varphi_{u,\theta}$ be a potential as in \Cref{prop:potentials-bounded}.

    Denote $k=|\alpha|$. By the multivariate Faà di Bruno's formula (see, e.g., Corollary 2.10 in \citet{constantine1996multivariate}),
    \begin{align}
        \nabla^\alpha \varphi_{u,\theta}(x)= \sum_{\beta_1+\dots + \beta_k =\alpha} \lambda_{\alpha,\beta_1,\dots,\beta_k} \prod_{j=1}^k \int e^{-\psi_{u,\theta}^0(y)} \nabla^{\beta_j} e^{-u'\phi(x,y,\theta)} \, d\nu(y), \label{eq:fdb1}
    \end{align}
    where the summation is over multi-indices $\beta_1,\dots,\beta_k \in \N_0^d$ such that $\beta_1+\dots+\beta_k=\alpha$ and $\lambda_{\alpha,\beta_1,\dots,\beta_k}$ are combinatorial quantities that only depend on $\alpha,\beta_1,\dots,\beta_k$.

    Now consider the expression $\nabla^\beta e^{-u'\phi(x,y,\theta)}$ for a multi-index $\beta$ of order $|\beta|\le s$. Again, by Faà di Bruno's formula,
    \begin{align*}
        \nabla^\beta e^{-u'\phi(x,y,\theta)} = \sum_{\gamma_1+\dots+\gamma_{|\beta|} = \beta} \nu_{\beta,\gamma_1,\dots,\gamma_{|\beta|}} \prod_{j=1}^{|\beta|} \nabla^{\gamma_j} [u'\phi(x,y,\theta)]
    \end{align*}
    for combinatorial quantities $\nu_{\beta,\gamma_1,\dots,\gamma_{|\beta|}}$ that only depend on $\beta,\gamma_1,\dots,\gamma_{|\beta|}$.
    Let $\preceq$ denote inequality up to a multiplicative constant that depends only on $s$ and $d$.
    Then
    \begin{align*}
        \left| \nabla^\beta e^{-u'\phi(x,y,\theta)} \right| \preceq \bar \phi_s^{|\beta|} \preceq 
        \max(\bar\phi_s, \bar \phi_s^{s}),
    \end{align*}
    where the last inequality holds because $|\beta|\le s$.

    Combining with \eqref{eq:fdb1} yields
    \begin{align*}
        \left|\nabla^\alpha \varphi_{u,\theta}(x)\right|  \preceq \sum_{\beta_1+\dots + \beta_k =\alpha} |\lambda_{\alpha,\beta_1,\dots,\beta_k}| \prod_{j=1}^k \bar \phi_s^s \int e^{-\psi_{u,\theta}^0(y)} \, d\nu(y) \preceq e^{ k \bar \phi} \bar \phi_s^{ks} \preceq e^{s \bar \phi} \max(\bar \phi_s^s, \bar \phi_s^{s^2}),
    \end{align*}
    where the last inequality holds because $k\le s$.
    The proof for the potential $\psi_{u,\theta}$ is analogous.
\end{proof}


The following proposition and its proof follow Lemma E.23 of \citet{goldfeld2024statistical}.

\begin{proposition}\label{prop:hadamard}
The mapping $\mu \mapsto c(\mu)$ is Lipschitz continuous and Gateaux differentiable, i.e., for any probability measures $\mu_0,\mu_1$ supported in $\mathcal{X}$, we have
\begin{align}
    \|c(\mu_1)-c(\mu_0)\|_{\B \times\Theta} \le \|\mu_1-\mu_0\|_{\F} \label{eq:c-Lipschitz}
\end{align}
and
\begin{align}
    \lim_{t\downarrow 0} \frac{c(\mu_0+t(\mu_1-\mu_0))(u,\theta)-c(\mu_0)(u,\theta)}{t} = \int \varphi_{u,\theta} \,d(\mu_1-\mu_0), \label{eq:c-gateaux}
\end{align}
where $\varphi_{u,\theta}$ is an optimal potential for $\mu_0$ and the cost function $u'\phi(\cdot,\cdot,\theta)$.
\end{proposition}

\begin{proof}
    We first prove \eqref{eq:c-gateaux}.
    Let $\mu_t = \mu_0+t(\mu_1-\mu_0)$ for $t\in[0,1]$ and let $\varphi_{u,\theta}^t$, $\psi_{u,\theta}^t$ be optimal potentials for $\mu_t$ and the cost function $u'\phi(\cdot,\cdot,\theta)$.
    Since $\mu_t$ is concentrated in $\mathcal{X}$, we can choose these potentials to satisfy the derivative bounds.
    Notice that
    \begin{align}
        \notag
        c(\mu_t)(u,\theta)&=\int \varphi_{u,\theta}^t \, d\mu_t + \int \psi_{u,\theta}^t \, d\nu \ge
        \int \varphi_{u,\theta}^0 \, d\mu_t + \int \psi_{u,\theta}^0 \, d\nu - \int e^{\varphi_{u,\theta}^0 \oplus \psi_{u,\theta}^0 - u'\phi(\cdot,\cdot,\theta) } \, d\mu_t \otimes \nu + 1\\
        \notag
        &= \int \varphi_{u,\theta}^0 \, d\mu_t + \int \psi_{u,\theta}^0 \, d\nu = \int \varphi_{u,\theta}^0 \, d\mu_0 + \int \psi_{u,\theta}^0 \, d\nu + t \int \varphi_{u,\theta}^0 \, d(\mu_1-\mu_0) \\
        &= c(\mu_0)(u,\theta) + t \int \varphi_{u,\theta}^0 \, d(\mu_1-\mu_0), \label{eq:theta-lower}
    \end{align}
    where the second equality uses the fact that $\int e^{\varphi_{u,\theta}^0(x)+\psi_{u,\theta}^0(y)-u'\phi(x,y,\theta)} d\nu(y)=1$ for all $x\in \mathcal{X}$.
    Hence,
    \begin{align*}
        \liminf_{t \downarrow 0} \frac{c(\mu_t)(u,\theta)-c(\mu_0)(u,\theta)}{t} \ge \int \varphi_{u,\theta}^0 \, d(\mu_1-\mu_0).
    \end{align*}
    Similarly, we have
    \begin{align}
        \notag
        c(\mu_t)(u,\theta)&=\int \varphi_{u,\theta}^t \, d\mu_t + \int \psi_{u,\theta}^t \, d\nu = \int \varphi_{u,\theta}^t \, d\mu_0 + \int \psi_{u,\theta}^t \, d\nu + t \int \varphi_{u,\theta}^t \, d(\mu_1-\mu_0) \\
        \notag
        &\le \int \varphi_{u,\theta}^0 \, d\mu_0 + \int \psi_{u,\theta}^0 \, d\nu + t \int \varphi_{u,\theta}^t \, d(\mu_1-\mu_0) + \int e^{\varphi_{u,\theta}^t \oplus \psi_{u,\theta}^t - u'\phi(\cdot,\cdot,\theta)} \,d\mu_0 \otimes \nu - 1 \\
        \notag
        &= \int \varphi_{u,\theta}^0 \, d\mu_0 + \int \psi_{u,\theta}^0 \, d\nu + t \int \varphi_{u,\theta}^t \, d(\mu_1-\mu_0) \\
        &= c(\mu_0)(u,\theta) + t \int \varphi_{u,\theta}^t \, d(\mu_1-\mu_0), \label{eq:theta-upper}
    \end{align}
    where the second equality uses the fact that $\int e^{\varphi_{u,\theta}^0(x)+\psi_{u,\theta}^0(y)-u'\phi(x,y,\theta)} d\mu_0(x)=1$ for all $y \in \mathcal{Y}$.
    Hence,
    \begin{align*}
        \frac{c(\mu_t)(u,\theta)-c(\mu_0)(u,\theta)}{t} \le \int \varphi_{u,\theta}^t \, d(\mu_1-\mu_0).
    \end{align*}
    It suffices to show that for any sequence $t_n \downarrow 0$,
    \begin{align}
        \lim_{n \to \infty} \int \varphi_{u,\theta}^{t_n} \, d(\mu_1-\mu_0) = \int \varphi_{u,\theta}^0 \, d(\mu_1-\mu_0). \label{eq:lim1}
    \end{align}
    Pick any subsequence $n'$ of $n$. From \eqref{eq:phi-deriv-bound} and Arzela-Ascoli theorem, there exists a further subsequence $n''$ such that $\varphi_{u,\theta}^{t_{n''}} \to \varphi_{u,\theta}$ and $\psi_{u,\theta}^{t_{n''}} \to \psi_{u,\theta}$ locally uniformly for some continuous functions $\varphi_{u,\theta},\psi_{u,\theta}$.
    Again, from \eqref{eq:phi-deriv-bound} and the dominated convergence theorem, $(\varphi_{u,\theta},\psi_{u,\theta})$ satisfy the first-order conditions, and hence they are optimal potentials for $(\mu_0,\nu)$ and the cost function $u'\phi(\cdot,\cdot,\theta)$.
    Let us now verify that $\varphi_{u,\theta}=\varphi_{u,\theta}^0+a$ for some constant $a\in \R$.
    By uniqueness of optimal potentials, $\psi_{u,\theta}=\psi_{u,\theta}^0-a$ for some $a\in \R$. Then,
    \begin{align*}
        \varphi_{u,\theta}(x)=-\log \int e^{\psi_{u,\theta}(y)-u'\phi(x,y,\theta)} \, d\nu(y)=-\log \int e^{\psi_{u,\theta}^0(y)-u'\phi(x,y,\theta)} \, d\nu(y) + a = \varphi_{u,\theta}^0(x)+a.
    \end{align*}
    Therefore, by \eqref{eq:phi-deriv-bound} and the dominated convergence theorem,
    \begin{align*}
        \lim_{n'' \to \infty} \int \varphi_{u,\theta}^{t_{n''}} \, d(\mu_1-\mu_0) = \int \varphi_{u,\theta} \, d(\mu_1-\mu_0) = \int \varphi_{u,\theta}^0\, d(\mu_1-\mu_0).
    \end{align*}
    Since the limit does not depend on the choice of the subsequence, this establishes \eqref{eq:lim1}, completing the proof of \eqref{eq:c-gateaux}.

    Next, we show \eqref{eq:c-Lipschitz}.
    From the inequalities \eqref{eq:theta-lower} and \eqref{eq:theta-upper}, we have
    \begin{align*}
        \left|c(\mu_1)(u,\theta)-c(\mu_0)(u,\theta) \right| \le \max\left( \int \varphi_{u,\theta}^0\, d(\mu_1-\mu_0), \, \int \varphi_{u,\theta}^1\, d(\mu_1-\mu_0)  \right) \le \|\mu_1-\mu_0\|_{\F},
    \end{align*}
    where the last inequality is due to $\varphi_{u,\theta}^0,\varphi_{u,\theta}^1 \in \F$.
    Since $\F$ is independent of $(u,\theta)$, the proof is completed.
\end{proof}

\begin{proposition}\label{prop:Donsker-joint}
    There exists a tight Gaussian process $\G_{\mu\otimes \nu}$ in $\ell^\infty(\mathscr{F}^\oplus)$ such that
    \begin{align*}
        \sqrt n(\hat \mu_n \otimes \hat\nu_n- \mu \otimes \nu) \weakto \G_{\mu\otimes \nu} \text{ in } \ell^\infty(\mathscr{F}^\oplus).
    \end{align*}
\end{proposition}
\begin{proof}
    See the proof of part (ii) of Theorem 1 in \citet{goldfeld2024statistical}.
\end{proof}

\subsection{Proof of \Cref{cor:D-distribution}}\label{app:D-distribution}

Theorem 3.1 in \citet{shapiro1991asymptotic} implies that the functional $\chi(c)=\max_{u\in\B} c(u)$ is Hadamard directionally differentiable with the derivative
\begin{align*}
    \chi_{c_{\theta,\varepsilon}}'(h) = \max_{u\in U_{c,\varepsilon}(\theta)} h(u), \quad h\in C(\B),
\end{align*}
where 
\begin{align*}
    U_{c,\varepsilon}(\theta) = \arg\max_{u\in \B} c_{\theta,\varepsilon}(u).
\end{align*}
Applying the functional delta method \citep[e.g.,][Theorem 3.10.5]{vaart2023empirical} to \Cref{thm:UCLT} completes the proof.

For consistency of the enlarged argmax and the estimated directional derivative, we follow Appendix F of \citet{franguridi2025set} extending the argument to a weaker nonlinear separation condition in \Cref{a:sharp-maxima}. Fix $\theta_0$ and $\varepsilon$, and let $e_n=\|\hat c_{\theta_0,\varepsilon}-c_{\theta_0,\varepsilon}\|_\infty=o_p(\iota_n)$ by \Cref{thm:UCLT} and \Cref{a:bandwidth}. Under \Cref{a:sharp-maxima}, the separation bound in their Lemma 1 becomes
\begin{align*}
    \rho\left(d(u,U_{c,\varepsilon}(\theta_0))\right)\le 2e_n+\iota_n=o_p(1), \quad u\in\hat U_{n,\varepsilon}.
\end{align*}
Since $\rho$ is nondecreasing and $\rho(\delta)>0$ for every $\delta>0$, this implies $\sup_{u\in\hat U_{n,\varepsilon}}d(u,U_{c,\varepsilon}(\theta_0))=o_p(1)$. The rest of Lemma 1 and the argument for consistency of the estimated directional derivative apply with maxima in place of minima.

\section{Details for panel logit with attrition and refreshment}\label{app:logit-details}

\subsection{Common slope parameter}\label{app:logit-beta}

For the fixed effects panel logit with $T=2$, the conditional log-likelihood for individuals with $S_i = Y_{i1}+Y_{i2}=1$ is
\begin{align*}
\ell_i(\theta \mid S_i=1) 
= Y_{i1}X_{i1}'\theta + Y_{i2}X_{i2}'\theta - \ln\left(e^{X_{i1}'\theta}+e^{X_{i2}'\theta}\right),
\end{align*}
with score
\begin{align*}
s(Z_{i1},Z_{i2},\theta) 
&= \left(Y_{i1}X_{i1}+Y_{i2}X_{i2}\right) 
- \frac{e^{X_{i1}'\theta}X_{i1}+e^{X_{i2}'\theta}X_{i2}}{e^{X_{i1}'\theta}+e^{X_{i2}'\theta}},
\end{align*}
and moment condition
\begin{align*}
\E\left[s(Z_{i1},Z_{i2},\theta_0) \mid S_i=1\right] = 0.
\end{align*}
We embed the event $\{Y_{i1}+Y_{i2}=1\}$ in the cost function
\begin{align*}
\phi\left(z_1,z_2,\theta\right) 
= s(z_1,z_2,\theta)\mathbf{1}\left\{y_1+y_2=1\right\}.
\end{align*}

We partition the population into retainers (observed in both periods) and attriters (observed only in period 1). This partitioned approach yields tighter bounds by fixing the known joint distribution of retainers and limiting the OT problem to the attriters only. Let $p$ denote the retention rate. For $t=1,2$, we use the following notation.
\begin{itemize}
\item $f_t(z)$: population distribution in period $t$,
\item $f_{t\mid\text{ret}}(z)$ and $f_{t\mid\text{att}}(z)$: distributions of retainers and attriters in period $t$,
\item $f_{1,2\mid\text{ret}}(z_1,z_2)$: joint distribution of retainers across both periods.
\end{itemize}
The refreshment sample is drawn from the same population, so $f_{2|\text{ref}}(z)=f_2(z)$. The attriter distribution $f_{2\mid\text{att}}(z)$ is unobserved.

By the law of total probability,
\begin{align*}
f_t &= p \cdot f_{t\mid\text{ret}} + (1-p) \cdot f_{t\mid\text{att}}, \quad t=1,2,
\end{align*}
so the unobserved attriter distribution in period 2 is
\begin{align*}
f_{2\mid\text{att}} 
= \frac{f_2 - p \cdot f_{2\mid\text{ret}}}{1-p}. 
\end{align*}

Let $\Pi(f_{1\mid\text{att}}, f_{2\mid\text{att}})$ denote the set of all joint distributions with these marginals. The directional OT criterion combines the contributions from retainers and attriters,
\begin{align*}
c_{\theta}(u)
&= p \cdot \E_{f_{1,2\mid\text{ret}}}\left[u'\phi(Z_1,Z_2,\theta)\right]
+ (1-p) \cdot \inf_{f \in \Pi(f_{1\mid\text{att}}, f_{2\mid\text{att}})} \E_f\left[u'\phi(Z_1,Z_2,\theta)\right], \quad u\in\B.
\end{align*}

\Cref{alg:fe-logit-bounds} presents the algorithm for computing bounds on the common slope parameter $\theta$. We estimate the marginal distributions by empirical cell frequencies. Alternative nonparametric estimators, such as kernel estimators, sieves, or splines, could be employed as well.
Let $z_{1i}$ and $z_{2j}$ index the marginal support cells.

\begin{algorithm}[t]
\caption{Panel logit with attrition and refreshment: common slope parameter}
\label{alg:fe-logit-bounds}
\begin{algorithmic}[1]
\Require Original panel at $t=1$: $\{Z_{i1}\}_{i=1}^{n_{\rm org}}$, retainers' sample at $t=2$: $\{Z_{i2}\}_{i\in\mathcal S}$, refreshment sample at $t=2$: $\{Z_{j2}\}_{j=1}^{n_{\rm ref}}$, regularization parameter $\varepsilon>0$, threshold $\eta_n$
\Ensure Estimated identified set $\widehat\Theta_I$ for parameter $\theta$

\smallskip\State $n_{\rm ret} \leftarrow |\mathcal S|$, $\hat p \leftarrow n_{\rm ret}/n_{\rm org}$

\smallskip\Statex \textit{Estimate empirical marginals:}
\State $\hat f_{1|{\rm att}}(z) \leftarrow \frac{1}{n_{\rm org}-n_{\rm ret}} \sum_{i\notin\mathcal S}\mathbf1\{Z_{i1}=z\}$
\State $\hat f_{2|{\rm ret}}(z) \leftarrow \frac{1}{n_{\rm ret}} \sum_{i\in\mathcal S}\mathbf1\{Z_{i2}=z\}$
\State $\hat f_{2}(z) \leftarrow \frac{1}{n_{\rm ref}} \sum_{j=1}^{n_{\rm ref}}\mathbf1\{Z_{j2}=z\}$

\smallskip\Statex \textit{Recover attriter marginal:}
\State $\hat f_{2|{\rm att}}(z) \leftarrow \frac{\hat f_{2}(z) - \hat p\hat f_{2|{\rm ret}}(z)}{1-\hat p}$

\smallskip\Statex \textit{Construct moments and solve entropic OT:}
\For{each candidate $\theta$}
    \State $\phi_{ij}(\theta) \leftarrow \phi(z_{1i},z_{2j},\theta)$ for all support pairs
    \State $\widehat M_{\rm ret}(\theta) \leftarrow \frac{1}{n_{\rm ret}} \sum_{i\in\mathcal S} \phi(Z_{i1},Z_{i2},\theta)$
    \State For $u\in\B$, compute
     $\widehat c^{\rm att}_{\theta,\varepsilon}(u) \leftarrow \min_{w\ge0}\left\{\sum_{i,j}w_{ij}u'\phi_{ij}(\theta)+\varepsilon\sum_{i,j}w_{ij}\log\frac{w_{ij}}{\hat f_{1|{\rm att}}(z_{1i})\hat f_{2|{\rm att}}(z_{2j})}\right\}$
     \Statex\hspace*{1.7em}subject to
     $\sum_j w_{ij}=\hat f_{1|{\rm att}}(z_{1i})$, $\sum_i w_{ij}=\hat f_{2|{\rm att}}(z_{2j})$.
    \State $\widehat c_{\theta,\varepsilon}(u) \leftarrow \hat p u'\widehat M_{\rm ret}(\theta)+(1-\hat p)\widehat c^{\rm att}_{\theta,\varepsilon}(u)$
    \State $\widehat D_\varepsilon(\theta) \leftarrow \max_{u\in\B}\widehat c_{\theta,\varepsilon}(u)$
    \State Compute the feasible upper bound $\widehat D_\varepsilon^{\mathrm{up}}(\theta)$ in \Cref{app:logit-bias}
\EndFor

\smallskip\Statex \textit{Construct estimated set:}
\State $\widehat\Theta_I \leftarrow \left\{\theta:\widehat D_\varepsilon^{\mathrm{up}}(\theta)\le\eta_n\right\}$
\smallskip\State \Return $\widehat\Theta_I$
\end{algorithmic}
\end{algorithm}

\subsection{AME}\label{app:logit-ame}

\subsubsection{Without attrition}\label{app:logit-ame-no-attrition} \citet{davezies2021identification} use Chebyshev polynomial approximation to construct outer bounds on the AME.\footnote{A sharper bound can be obtained via Hankel moment matrix positivity, but it involves nonparametric first-step estimation.}
Let $X = (X_1, \ldots, X_T)$ denote the full stack of period-specific covariates and $S_i = \sum_{t=1}^T Y_{it}$. The AME of covariate $j$ at period $\tau$ is
\[
\delta_{\tau,j} 
= \theta_j\E[\Lambda(X_{\tau}'\theta + \alpha)(1-\Lambda(X_{\tau}'\theta + \alpha))].
\]
The outer bounds of $\delta_{\tau,j}$ are given by $\tilde{\delta} \pm \bar{b}$, where $\tilde{\delta} = \E[p(X, S, \theta_0)]$ and $\bar{b} = \E[a(X,S,\theta_0)]$, with
\begin{align*}
p(x, s, \theta) &= \sum_{t=0}^{s} \left(\lambda_t(x, \theta) + b^*_{t,T} \lambda_{T+1}(x, \theta)\right) \binom{T-t}{s-t} \frac{\exp(sx_{\tau}' \theta)}{C_s(x, \theta)}, \\
a(x,s,\theta) &= \frac{1}{2 \times 4^T}  |\lambda_{T+1}(x, \theta)| \binom{T}{s} \frac{\exp(sx_{\tau}'\theta)}{C_s(x, \theta)}, \\
C_k(x, \theta) &= \sum_{(d_1,\ldots,d_T) \in \{0,1\}^T: \sum_{t=1}^T d_t = k} \exp\left(\sum_{t=1}^T d_t x_t' \theta\right).
\end{align*}
For notational simplicity, we suppress the subscripts $\tau$ and $j$ when there is no ambiguity. Here $b^*_{t,T}$ are the Chebyshev coefficients that are fixed constants fully determined by the degree-$(T+1)$ Chebyshev polynomial rescaled to $[0,1]$, independent of any data. The functions $\lambda_t(x,\theta)$ are defined as the monomial coefficients in
\begin{align*}
\sum_{t=0}^{T+1} \lambda_t(x,\theta) u^t 
= \theta_j u(1-u) \prod_{t\neq\tau} \left( 1 + u \left\{ \exp\left(\left(x_t - x_\tau\right)'\theta\right) - 1 \right\} \right).
\end{align*}

For simplicity, we now illustrate the case with $T=2$. Then
\begin{align*}
C_0(x, \theta) = \exp(0) = 1, \quad
C_1(x, \theta) = \exp(x_1' \theta) + \exp(x_2' \theta), \quad
C_2(x, \theta) = \exp(x_1' \theta + x_2' \theta).
\end{align*}
For $\lambda_t(x,\theta)$, w.l.o.g.\ let $\tau=1$. Due to the factor $u(1-u)$, we have $\lambda_0(x,\theta)=0$. Then, expanding the defining identity gives
\begin{align*}
  \lambda_1(x,\theta) u + \lambda_2(x,\theta) u^2 + \lambda_3(x,\theta) u^3 
  = \theta_j u(1-u) \left( 1 + u \left\{ \exp\left(\left(x_2 - x_1\right)'\theta\right) - 1 \right\} \right),
\end{align*}
which implies
\[
 \lambda_1(x,\theta)=\theta_j, \quad 
 \lambda_2(x,\theta)=\theta_j\left(\exp\left(\left(x_2-x_1\right)'\theta\right)-2\right), \quad
 \lambda_3(x,\theta)=\theta_j\left(1-\exp\left(\left(x_2-x_1\right)'\theta\right)\right).
\]

Using these explicit forms, we can simplify the expressions for $p(x,s,\theta)$ and $a(x,s,\theta)$. For $S=0$, 
\begin{align*}
p(x,0,\theta)&=  b^*_{0,2} \lambda_3(x,\theta),\quad
a(x,0,\theta) = \frac{1}{32} |\lambda_{3}(x, \theta)|.
\end{align*}
For $S=1$,
\begin{align*}
p(x,1,\theta) &= \left[\theta_j+(2b^*_{0,2}+ b^*_{1,2})\lambda_3(x,\theta) \right] \frac{\exp(x_{\tau}'\theta)}{\exp(x_1'\theta) + \exp(x_2'\theta)},\\ 
a(x,1,\theta) &= \frac{1}{16} |\lambda_{3}(x, \theta)|  \frac{\exp(x_{\tau}'\theta)}{\exp(x_1'\theta) + \exp(x_2'\theta)}.
\end{align*}
For $S=2$,
\begin{align*}
p(x,2,\theta) &= ( b^*_{0,2}+ b^*_{1,2}+ b^*_{2,2}-1)\lambda_3(x,\theta)  \frac{\exp(2x_{\tau}'\theta)}{\exp(x_1'\theta + x_2'\theta)},\\ 
a(x,2,\theta) &= \frac{1}{32} |\lambda_{3}(x, \theta)|  \frac{\exp(2x_{\tau}'\theta)}{\exp(x_1'\theta + x_2'\theta)}.
\end{align*}
By plugging in $\hat\theta$ and forming the corresponding sample analogues, one obtains the estimates $\tilde\delta$ and $\bar b$ used in the AME bounds in the main text. \citet{davezies2021identification} also show how to construct valid confidence regions based on these bounds.

\subsubsection{With unrestricted attrition}\label{app:logit-ame-attrition}
To estimate the identified set for the AME under unrestricted attrition, we proceed in three steps.

\textbf{Step 1: identified set for $\theta$.}
Compute the identified set $\widehat\Theta_I$ for the common slope parameters $\theta$ as described above and in \Cref{alg:fe-logit-bounds}. Let $\left\{\theta^{(g)}\right\}_{g=1}^G$ be a finite grid covering $\widehat\Theta_I$.

\textbf{Step 2: identified set for the AME conditional on $\theta^{(g)}$.}
For each grid point, plug in $\theta^{(g)}$ into the Chebyshev polynomial approximation and define cost functions
\[
\left[\underline\phi(x,s,\theta^{(g)}),\, \overline\phi(x,s,\theta^{(g)})\right] = p(x,s,\theta^{(g)}) \pm a(x,s,\theta^{(g)}),
\]
where $p$ and $a$ are as in \Cref{app:logit-details}. The identified set for the AME under $\theta^{(g)}$ is $\left[\underline{\delta}(\theta^{(g)}), \overline{\delta}(\theta^{(g)})\right]$, where we combine retainers and attriters as
\begin{align*}
\underline{\delta}(\theta^{(g)}) 
&= p \cdot \E_{\text{ret}}\left[\underline\phi(X,S, \theta^{(g)})\right] + (1-p) \cdot \inf_{f \in \Pi(f_{1\mid\text{att}}, f_{2\mid\text{att}})} \E_f\left[\underline\phi(X,S, \theta^{(g)})\right], \\
\overline{\delta}(\theta^{(g)}) 
&= p \cdot \E_{\text{ret}}\left[\overline\phi(X,S, \theta^{(g)})\right] + (1-p) \cdot \sup_{f \in \Pi(f_{1\mid\text{att}}, f_{2\mid\text{att}})} \E_f\left[\overline\phi(X,S, \theta^{(g)})\right].
\end{align*}
The retainer expectations are computed directly from the observed data.

Since $0\leq\Lambda(v)(1-\Lambda(v))\leq1/4$, we intersect these bounds with the logit derivative restriction before profiling:
\[
\left[\underline{\delta}(\theta),\overline{\delta}(\theta)\right]
\cap
\left[\min\{0,\theta_j/4\},\max\{0,\theta_j/4\}\right].
\]

\textbf{Step 3: profiling over $\theta$.}
Finally, we take the union over the grid to obtain the identified set for the AME
\[
\bigcup_{g=1}^G \left[\underline{\delta}(\theta^{(g)}), \overline{\delta}(\theta^{(g)})\right].
\]

\subsection{Regularization bias and tuning parameters}\label{app:logit-bias}

Write $\mu=f_{1\mid\mathrm{att}}$, $\nu=f_{2\mid\mathrm{att}}$, and $M_{\rm ret}(\theta)=\E_{f_{1,2\mid\mathrm{ret}}}\phi(Z_1,Z_2,\theta)$. In our simulations and the empirical application, $Z_t$ has one continuous covariate and several discrete variables.
The regularized criterion function is $D_\varepsilon(\theta)=\max_{u\in\B}c_{\theta,\varepsilon}(u)$, where
\begin{align*}
c_{\theta,\varepsilon}(u)
&=p u'M_{\rm ret}(\theta)+(1-p)\inf_{\pi\in\Pi(\mu,\nu)}
\left\{\E_\pi[u'\phi(Z_1,Z_2,\theta)]+\varepsilon\operatorname{KL}(\pi\|\mu\otimes\nu)\right\}.
\end{align*}
A feasible coupling $\pi\in\Pi(\mu,\nu)$ with finite relative entropy gives an upper bound $D_\varepsilon^{\mathrm{up}}(\theta)$ on $D_\varepsilon(\theta)$. For any $u\in\B$, any pair of potentials satisfying $\varphi(z_1)+\psi(z_2)\leq u'\phi(z_1,z_2,\theta)$ gives a lower bound $L_0(\theta)$ on $D(\theta)$, where
\begin{align*}
D_\varepsilon^{\mathrm{up}}(\theta)
&=\left\|pM_{\rm ret}(\theta)+(1-p)\E_\pi\phi(Z_1,Z_2,\theta)\right\|
+(1-p)\varepsilon\operatorname{KL}(\pi\|\mu\otimes\nu),\\
L_0(\theta)
&=\max\left\{0,\,\, p u'M_{\rm ret}(\theta)+(1-p)\left(\int\varphi d\mu+\int\psi d\nu\right)\right\}.
\end{align*}
Therefore, $r(\varepsilon,\theta)=D_\varepsilon^{\mathrm{up}}(\theta)-L_0(\theta)$ satisfies
\[
0\leq D_\varepsilon(\theta)-D(\theta)\leq r(\varepsilon,\theta).
\]
For fixed $\varepsilon>0$, the test based on the corrected statistic still controls size locally uniformly since under the null $\sqrt n\{\widehat D_\varepsilon(\theta)-r(\varepsilon,\theta)\}\leq\sqrt n\{\widehat D_\varepsilon(\theta)-D_\varepsilon(\theta)\}$.

In practice, we compute the corresponding sample bounds $\widehat D_\varepsilon^{\mathrm{up}}\geq\widehat D_\varepsilon$ and $0\leq\widehat L_0\leq\widehat D$, and set $\widehat r(\varepsilon,\theta)=\widehat D_\varepsilon^{\mathrm{up}}(\theta)-\widehat L_0(\theta)$. We obtain $\widehat D_\varepsilon^{\mathrm{up}}$ by adjusting the OT plan to satisfy the estimated marginals, and compute $\widehat L_0$ from potentials satisfying the unregularized dual constraints.
Let $\widehat D_\varepsilon^{\rm low}$ denote the attained regularized dual lower bound. Then
\[
\widehat D_\varepsilon^{\rm low}(\theta)-\widehat r(\varepsilon,\theta)
\leq\widehat D_\varepsilon(\theta)-\widehat r(\varepsilon,\theta)
\leq\widehat L_0(\theta)\leq\widehat D(\theta),
\]
and the test rejects when
\[
\sqrt n\left\{\widehat D_\varepsilon^{\rm low}(\theta)-\widehat r(\varepsilon,\theta)\right\}-\kappa_n
>\widehat q_{\alpha,\varepsilon}^*(\theta),
\]
where $n=n_{\rm org}$ and $\widehat q_{\alpha,\varepsilon}^*(\theta)$ is the critical value obtained using \Cref{alg:bootstrap-test}. In the bootstrap procedure, the panel of stayers and the refreshment sample are resampled independently.
The term $\kappa_n$ is included to allow for the sampling error, which makes the test slightly more conservative.
We set $\kappa_n/\sqrt n=c_\kappa b_{n,d}\sqrt{\log n}$, where $b_{n,d}$ is defined in \Cref{thm:eot-stability}. In our simulations and the empirical application, the constant is chosen to be $c_\kappa=0.004$.

\end{subappendices}

\end{document}